\documentclass[a4paper, 12pt]{article}
\usepackage{bbm}
\usepackage{array,arydshln}
\usepackage{comment}
\usepackage{xcolor,soul}

\usepackage[round]{natbib}
\usepackage{hyperref}

\usepackage{graphicx, graphics}
\usepackage{amsmath,amsthm,amssymb,setspace,bm}

\newcommand{\KK}[1]{\textcolor{red}{#1}}

\newtheorem{theorem}{Theorem}
\newtheorem*{llt}{LLT}
\newtheorem{prop}{Proposition}

\newtheorem{claima}{Claim}
\newtheorem{lem}{Lemma}

\theoremstyle{definition}
\newtheorem{rem}{Remark}
\newtheorem{definition}{Definition}
\newtheorem{ex}{Example}
\newtheorem*{ex_rules}{Examples of rules}
\newtheorem{as}{Assumption}

\usepackage[margin=1.8in, left=1.2in, right=1.2in]{geometry}
\begin{document}

\title{The Winner-Take-All Dilemma\footnote{
	For their fruitful
discussions and useful comments, we are thankful to Pierre Boyer,  
Micael Castanheira, 
Quoc-Anh Do, 
Shuhei Kitamura,
Michel Le Breton,
Takehito Masuda,
Shintaro Miura,
Takeshi Murooka, 
Mat\'{i}as N\'{u}\~{n}ez, 
Alessandro Riboni, 
and 
the seminar participants at 
the 11th Pan Pacific Game Theory Conference, 
CREST, 
Hitotsubashi Summer Institute, 
Institut Henri Poincar\'{e}, 
Institute of Social and Economic Research at Osaka University, 
Osaka School of International Public Policy Lunch Seminar,
Paris School of Economics,
Parisian Political Economy Workshop, 
Philipps-University of Marburg,
Transatlantic Theory Workshop,
University of Montpelier,
University of Rochester,
University of Tokyo,
Waseda University
and the 2019 Autumn Meeting of the Japanese Economic Association.
Financial support by Investissements d'Avenir, ANR-11-IDEX-0003/Labex Ecodec/ANR-11-LABX-0047 and DynaMITE: Dynamic Matching and Interactions: Theory and Experiments, ANR-13-BSHS1-0010,
PHC Sakura program, project number 45153XK,
Waseda University Grant for Special Research Projects (2018K-014), 
JSPS KAKENHI Grant Numbers JP17K13706 and JP15H05728 is gratefully acknowledged.
}}
\author{
	Kazuya Kikuchi\thanks{
	{Tokyo University of Foreign Studies. E-mail: \texttt{kazuya.kikuchi68@gmail.com}}.}
	\and
	Yukio Koriyama\thanks{
	CREST, Ecole Polytechnique, Institut Polytechnique de Paris. %, University of Paris-Saclay. 
	E-mail: \texttt{yukio.koriyama@polytechnique.edu}.}}  
\maketitle
\begin{abstract}
We consider collective decision making when the society consists of groups endowed with voting weights. Each group chooses an internal rule that specifies the allocation of its weight to the alternatives as a function of its members' preferences. Under fairly general conditions, we show that the winner-take-all rule is a dominant strategy, while the equilibrium is Pareto dominated, highlighting the dilemma structure between optimality for each group and for the whole society. We also develop a technique for asymptotic analysis and show Pareto dominance of the proportional rule. Our numerical computation for the US Electoral College verifies its sensibility.
\end{abstract}

\section{Introduction}

In many situations of collective decision making, including representative democracy, the society consists of distinct groups and decisions are made based on opinions aggregated within the groups.
%In representative democracy, the representatives make social decisions based on aggregated opinions of the individuals.
For example, in the United States presidential election, each state allocates its electoral votes based on the statewide popular vote. Another example is legislative voting, in which each party indicates how legislators should vote based on the opinions of party members.

In such situations, the social decision depends on the rules that groups use to aggregate the opinions of their members. 
However, if groups choose their rules based on private motives, the resulting social decisions may not be desirable.
The resulting decisions may even make all groups worse off than they could have been.
This paper studies the relationship between groups' incentives and their welfare consequences.

Existing institutions use a variety of rules, many of which pertain to how to allocate the weight assigned to each group. 
On the one hand, the \textit{winner-take-all rule} devotes all the weight to the alternative preferred by the majority of its members.  
This rule has been used to allocate electoral votes in all but two states in the most recent US presidential elections.
A council of national ministers, each with a weighted vote (e.g., the Council of the European Union), is another example, provided that the ministers represent the majority interests of their country. 
Party discipline, frequently observed in legislative voting, is also an example of the winner-take-all rule used by parties.

On the other hand, the \textit{proportional rule} allocates a group's weight in proportion to the number of members who prefer the respective alternatives. 
In a wide range of parliamentary institutions at the regional, national and international levels, each group (e.g., constituency, prefecture or state) elects a set of representatives whose composition proportionally reflects the preferences of its citizens. 
Alternatively, when representatives are considered to represent parties rather than states or prefectures, the proportional rule corresponds to a party's rule that allows its representatives to vote for or against proposals according to their own preferences, provided that the composition of the party's representatives proportionally reflects the opinions of all party members.

Weight allocation rules are often exogenously given to all groups, but there are also cases where each group chooses its own rule.
For instance, in national parliaments, how the representatives are elected in the respective constituencies is stipulated by national law.
In contrast, parties often have control over how their representatives vote, by punishing those who violate the party line.
As another example, the US Constitution stipulates that it is up to each state to decide the way in which the presidential electoral votes are allocated (Article II, Section 1, Clause 2).

If groups are allowed to choose their own rules, each group may have an incentive to allocate the weight so as to increase the influence of its members' opinions on social decisions, at the expense of the influence of other groups. It is not clear whether such a group-level incentive is consistent with desirable properties of overall preference aggregation, such as Pareto efficiency. A society consisting of distinct groups thus faces a dilemma between the private incentive of each group and the overall social objectives. To address this issue, we model the choice of rules as a non-cooperative game.

In this paper, we consider a model of collective decision making where a society consists of groups endowed with voting weights. 
Each group chooses the rule for allocating its weight to the binary alternatives, and the winner is the one with the most weight. 
A \textit{rule} for a group is a function that maps members' preferences (e.g., group-wide popular vote) to an allocation of the weight to the alternatives. {Any Borel-measurable function is allowed, including the winner-take-all and proportional rules stated above.} 
A \textit{profile} is a specification of the rules for all groups. 
We study the game in which the groups independently choose their rules, so as to maximize the expected welfare of their members.

The main result of this paper is that the game is an {$n$-player} \textit{Prisoner's Dilemma} (Theorem \ref{thm:dilemma}).
The winner-take-all rule is a \textit{dominant strategy}, {i.e.}, it is an optimal strategy for each group, regardless of the rules chosen by the other groups.
However, if each group has less than half of the total weight, then the winner-take-all profile is \textit{Pareto dominated}, {i.e.}, another profile makes \textit{every} group better off.
In brief, no group has an incentive to deviate from the winner-take-all rule, but every group would be better off if all groups jointly moved to another profile.
The dilemma structure exists for any number of groups ($>$2) and with fairly little restriction on the joint distribution of preferences (Assumption \ref{as:general}).
{Members' preferences are allowed to be biased and correlated within and across groups.
For example, the model can be applied to parties with distinct but overlapping political goals, or to states with different levels of support for specific alternatives, such as blue, red or swing states in the US elections.}

The observation that the winner-take-all rule is an optimal strategy for groups is not new. 
As we will discuss in detail in Sections \ref{subsec:literature} and \ref{subsec:result}, previous studies have already pointed out such incentives for groups in various voting situations.
The theoretical prediction about group behavior is also consistent with the fact that {it has been dominantly employed by the states in the US Electoral College since the 1830s in order to allocate presidential electoral votes,} and also with the party discipline behaviors widely observed in assemblies.
Despite the various problems or limitations that have been pointed out concerning the winner-take-all rule,\footnote{
	There are multiple arguments against the winner-take-all rule.
	First, the winner of the election may be inconsistent with that of the popular vote (\cite{May1948}, \cite{Feix2004}).
	Such a discrepancy has occurred five times in the history of the US presidential elections, including recently in 2000 and 2016. Second, it may cause reduced dimensionality: (i) the parties have an incentive to concentrate campaign resources only in the battleground states, and 
	(ii) the voters' incentive to turn out or to invest in information may be weak and/or uneven across states, since the probability of each voter being pivotal is so small under the winner-take-all rule, and even smaller in non-swing states. {Although campaign resource allocation and voter turnout are important issues, they are beyond the scope of this paper.}}
it is still used prevalently.\footnote{
	A recent attempt of reform took place in 2004 in Colorado, when a ballot initiative for a state constitution amendment was raised, proposing the proportional rule. 
	The amendment failed to pass, garnering only 34.1\%  
	approval.
}

%Pareto inefficiency
%of the winner-take-all profile may appear
%to be a trivial consequence
%of the fact that direct majority
%voting by all individuals
%maximizes the utilitarian welfare of the society.\footnote{This property of majority rule holds as long as preference
%intensities (or welfare weights) are the same for all individuals.}
%In our model, the utilitarian property of majority rule
%translates into the fact that
%if the weights are proportional to
%the group sizes,
%the proportional profile
%maximizes the average expected welfare of 
%all individuals.
%However,
%this does not
%mean that \textit{every} group will be worse off
%under the winner-take-all profile,
%compared to the proportional (or any other)
%profile.
%As a simple counterexample,
%if some group has
%more than half the total weight,
%then the winner-take-all profile is optimal for that group,
%and hence
%Pareto efficient.
%This may suggest the possibility that
%even if no group
%has more than half the
%total weight,
%the winner-take-all profile 
%can still be Pareto efficient
%by benefiting
%groups with sufficiently
%large weights
%(e.g.,
%large states in the US such as California).
%Our result negates this possibility.
%In addition,
%there are also cases
%where
%the proportional profile does not
%Pareto dominate
%the winner-take-all profile,
%as the latter makes a group with a small weight
%better off (Example \ref{ex:3groups}).

The main contribution of this paper is to establish that, under quite general circumstances, the winner-take-all profile is Pareto dominated, i.e., \textit{every} group would be better off if all groups simultaneously changed their rules. 
This point should be distinguished from the conventional wisdom that direct popular vote (i.e., majority voting by all individuals) maximizes the \textit{utilitarian} welfare of the society, as it maintains the possibility that \textit{some} groups (e.g., small states) may benefit from the winner-take-all profile.
We provide a counterexample in Example \ref{ex:3groups}: a small group is strictly better off under the winner-take-all profile compared with both the direct popular vote and the proportional profile.
Indeed, protecting the interests of minority states is an oft-used argument by advocates of the Electoral College and its adherence to the winner-take-all rule. 
The welfare criterion used in Theorem \ref{thm:dilemma} is Pareto dominance, which is obviously stronger than the utilitarian welfare evaluation: there exists a profile under which \textit{every} group is better off than the winner-take-all profile. 
Example \ref{ex:3groups} shows that the dominant profile is not necessarily the proportional one nor the direct popular vote. If such is the case, what profile Pareto dominates the winner-take-all profile? A full characterization of the Pareto set is provided in Lemma \ref{lem:char}.

To further investigate welfare properties, we turn to an asymptotic and normative analysis of the model.
We consider situations where the number of groups is sufficiently large, and the preferences are independent across groups and distributed symmetrically with respect to the alternatives.
Under these conditions, we show that the proportional profile \textit{Pareto dominates} every other symmetric profile ({i.e.}, one in which all groups use the same rule), including the winner-take-all profile.
The assumptions on the preference distribution abstract from the fact that some groups may prefer specific alternatives.
Such an abstraction would be reasonable on the grounds that normative judgment about rules should not favor particular groups because of their characteristic preference biases.
To see how many groups are typically sufficient for the asymptotic result, we {provide} numerical computations in a model based on the US Electoral College, using the current apportionment of electoral votes.
The numerical comparisons indicate that the proportional profile does Pareto dominate the winner-take-all profile in the model with fifty states and one federal district.

While the above result suggests that the proportional profile asymptotically performs well in terms of efficiency, it is silent about the equality of individuals' welfare.
We apply our model to study how rules affect the distribution of welfare, by examining an asymmetric profile called the \textit{congressional district profile}.
This profile is inspired by the Congressional District Method currently used by Maine and Nebraska, in which two electoral votes are allocated by the winner-take-all rule, and the remainder are awarded to the winner of the popular vote in each district.\footnote{
	The idea of allocating a portion of the votes by the winner-take-all rule and allowing the rest to be awarded to distinct candidates can be seen as a compromise between the winner-take-all and the proportional rules. Symbolically, the two votes allocated by the winner-take-all rule correspond to the number of the Senators from each state, while the remainder is equal to the number of the House representatives. The idea behind such a mixture is in line with the logic supporting bicameralism, which is supposed to provide checks and balances between the state autonomy and federal governance. 
}
We show that the congressional district profile achieves a more equal distribution of welfare than any symmetric profile by making individuals in
smaller groups better off.

A technical contribution of this paper is to develop an asymptotic method for analyzing the expected welfare of players in weighted voting games. One of the major challenges in the analysis of these games is their discreteness. Due to the nature of combinatorial problems, obtaining an analytical result often requires a large number of classifications by cases, which may include prohibitively tedious and complex tasks in order to obtain general insights. We overcome this difficulty by considering asymptotic properties of games in which there are a sufficiently large number of groups.
This technique allows us to obtain an explicit formula that captures the asymptotic behavior of the payoffs, which is valid for a wide class of weight distributions among groups (the correlation lemma: Lemma \ref{lem:pareto}).

\subsection{Literature Review} \label{subsec:literature}

The incentives for groups to use the winner-take-all rule have been studied in several papers.
\citet{Hummel2011} and \citet{BeisbartBovens2008} analyze models of the US presidential elections.
\citet{Gelman2003} and \citet{Eguia2011GEB, Eguia2011AJPS} provide theoretical explanations as to why voters in an assembly form parties or voting blocs to coordinate their votes.
Their findings are coherent with our observation that the winner-take-all rule is a dominant strategy.
In particular, \citet{BeisbartBovens2008} and \citet{Gelman2003} compare the winner-take-all and proportional profiles.
In the context of the current apportionment of electoral votes in the US, \citet{BeisbartBovens2008} numerically compares these profiles, in terms of inequality indices on citizens' voting power and the mean majority deficit, on the basis of \textit{a priori} and \textit{a posteriori} voting power measures.
\citet{Gelman2003} compares the case with coalitions of equal sizes in which voters coordinate their votes to the case without such coordination.
Our analysis is based on Pareto dominance between profiles, and provides results which hold under a general distribution of weights or group sizes.
In that sense, the positive analysis of \citet{BeisbartBovens2008} is complementary to our normative analysis of the properties of the proportional profile.

We take groups as given. In the context of voting in an assembly, this means that we focus on the situation after the formation of parties or voting blocs, which is complementary to \citet{Eguia2011GEB, Eguia2011AJPS} that study the endogenous formation of such groups.
%The exogenous treatment of groups allows us to draw clear welfare implications from the model. 
%It would be interesting to extend the analysis in this paper to a more comprehensive model that includes endogenous group formation.

\citet{DeMouzon2019} provides a welfare analysis of popular vote interstate compacts, and shows that, for a regional compact, the welfare of member states is single-peaked as a function of the number of participating states, while it is monotonically decreasing for non-member states. The second effect dominates in terms of social welfare, unless a large majority (approximately more than $2/\pi \simeq 64\%$) of the states join the compact, implying that a small- or medium-sized regional compact is welfare detrimental. For a national compact, the total welfare is increasing, as it turns out that even non-members would mostly benefit from the compact, implying that the social optimum is attained when a majority joins the compact, {i.e.}, the winner is determined by the national popular vote. Their findings are coherent with ours: if the winner-take-all rule is applied only to a subset of the groups, then the member states enjoy the benefit at the expense of the welfare loss of the non-member states, and the total welfare decreases.
The social optimum is attained when the entire nation uses the popular vote.
The possibility of the national popular vote as a coordination device is discussed also in \citet{CloleryKoriyama2020}.

The history, objectives, problems, and reforms of the US Electoral College are summarized, for example, in \citet{Edwards2004}, \citet{BughEd2010} and \citet{Wegman2020}.
One of the most commonly discussed problems of the Electoral College is its reduced dimensionality. 
The incentive for the candidates to concentrate their campaign resources in swing and decisive states is modeled in \citet{Stromberg2008}, which is coherent with the findings of the seminal paper in probabilistic voting by \citet{LindbeckWeibull1987}. \citet{Stromberg2008} also finds that uneven resource allocation and unfavorable treatment of minority states would be mitigated by implementing a national popular vote, consistent with the classical findings by \citet{BramsDavis1974}. 
The incentive of voters to turn out is investigated by \citet{Kartal2015EJ}, which finds that the winner-take-all rule discourages turnout when the voting cost is heterogeneous.

Constitutional design of weighted voting is studied extensively in the literature. Seminal contributions are found in the context of power measurement: \citet{Penrose1946}, \citet{ShapleyShubik1954}, \citet{Banzhaf1968} and \citet{Rae1969}. Excellent summaries of theory and applications of power measurement are given by, above all, \citet{FelsenthalMachover1998} and \citet{LaruelleValenciano2008}. The tools and insights obtained in the power measurement literature 
are often used in the apportionment problem: e.g., \citet{BarberaJackson2006}, \citet{Koriyamaetal2013}, and \citet{Kurz2017}. 
Our analysis can be interpreted in the context of Bayesian mechanism design, by considering the groups in our model as agents whose preference intensity is private information. In this interpretation, Theorem \ref{thm:dilemma} translates into an impossibility theorem which states that there is no social choice function that is Bayesian incentive compatible, Pareto efficient and non-dictatorial. 
This is consistent with the results obtained in previous papers on Bayesian mechanism design, such as \citet{BORGERS20092057}, \citet{AzrieliKim2014} and \citet{EHLERS202031}.
The precise statement of the impossibility theorem (Proposition \ref{prop:mech}) and a discussion of the mechanism design literature will appear in Subsection \ref{sec:imp_mech}.

\begin{comment}
Theoretical analysis of unilateral deviation by a state from the winner-take-all to proportional rule has been studied by several papers in the literature. \citet{Hummel2011} shows that a majority of voters in the deviating state are worse off. He also provides sufficient conditions on the number of electoral voters and preferences, under which the voters in the other states are better off. \citet{BeisbartBovens2008} analyze Colorado's attempt in 2004 of reforming the winner-take-all rule to the proportional rule. They show that the unilateral deviation from the winner-take-all rule is not profitable for the state, in terms of both a priori measure based on the Banzhaf-Penrose Index, and a posteriori measure based on the sentiment indices, estimated by the data in the ten elections from 1968 to 2004. 
\end{comment}

\section{The Model}
\label{sec:themodel}

We consider collective decision making
when 
a society consists of 
groups
endowed with voting weights.
We first describe
the weighted voting
mechanism
(Section \ref{sec:indirectvoting}).
We then
construct a
non-cooperative
game in which
each group chooses an internal rule
that specifies
the allocation
of its weight to the alternatives as a function of its members' preferences (Section \ref{sec:gamma}).
Finally,
we introduce
social choice functions 
which include
the weighted voting
mechanism as a special case  (Section \ref{subsec:scf}).
%We will use these functions to discuss the implications of our main result (Theorem \ref{thm:dilemma}) to general mechanism design problems (Section \ref{sec:imp_mech}).

\subsection{Weighted Voting}
\label{sec:indirectvoting}

Let us begin with the description
of the social decision process.
We consider a society
partitioned into
$n$ disjoint groups:
$i\in\{1,2,\cdots,n\}$.
Each group $i$
is endowed with a voting 
weight $w_i>0$.

The society makes a decision
between two
alternatives,
denoted
$-1$ and $+1$,
through the following
two voting stages:
(i) each individual
votes for his
preferred alternative;
(ii) each group allocates
its weight between the
alternatives, based on
the group-wide voting result.
The winner is the alternative
that receives the majority
of overall weight.

Let
$\theta_i\in[-1,1]$
denote the
vote margin in group $i$
at the first voting stage.
That is,
$\theta_i$
is the fraction
of members of $i$
preferring alternative
$+1$
minus the fraction
preferring $-1$.\footnote{For example, $\theta_i=0.2$
means that
$60\%$ of members
of $i$ prefer $+1$
and $40\%$ prefer $-1$.}

At the second stage,
each group's allocation of weight is determined as a function of the group-wide margin.
\begin{definition}
	A \textit{rule}
	for group $i$
	is defined as a Borel-measurable\footnote{Borel-measurability
		is needed to ensure that
		$\phi_i(\theta_i)$
		is a well-defined random variable, when $\theta_i$ is a random variable.} function:
	\[
	\phi_i:[-1,1]\to[-1,1].
	\]
	\label{def:rule}
\end{definition}
The value $\phi_i(\theta_i)$
is the group-wide weight margin,
{i.e.}, the
fraction
of the weight $w_i$
allocated to
alternative $+1$
minus that allocated to $-1$,
given that the
vote margin is $\theta_i$.
That is,
the rule allocates
$w_i\phi_i(\theta_i)$
more weight to
alternative
$+1$
than alternative 
$-1$.\footnote{For example, if $w_i=50$ and $\phi_i(\theta_i)=0.2$, it means that the rule allocates $30$ (resp. $20$) units of weight to the alternative $+1$ (resp. $-1$) so that the weight margin in favor of the alternative $+1$ is $50 \times 0.2 = 30- 20$.}

Let 
$$
\Phi = \left\{ \phi_i | \text{Borel-measurable} \right\}
$$
be the set of all admissible rules.
\begin{ex_rules}
	Among all admissible rules, the following examples deserve particular attention. 
	\begin{itemize}
		\item[(i)]
		\textit{Winner-take-all rule}:
		$\phi_i^{\rm WTA}(\theta_i)=
		\text{sgn}\,\theta_i$.
		\item[(ii)]
		\textit{Proportional
			rule}:
		$\phi_i^{\rm PR}(\theta_i)=\theta_i$.
		\item[(iii)]
		\textit{Mixed rules}:
		$\phi_i^a(\theta_i)=a \phi^{\rm WTA}_i(\theta_i)+(1-a)\phi^{\rm PR}_i(\theta_i)$, $0\leq a\leq1$.
	\end{itemize}
	The winner-take-all rule
	devotes
	all the weight of a group
	to the winning alternative in the group.
	The proportional rule
	allocates
	the weight in proportion
	to the vote shares of the respective alternatives
	in the group.
	The mixed
	rule $\phi^a$
	allocates
	the fixed ratio
	$a$
	of the weight
	by the winner-take-all rule
	and the remaining
	$1-a$ part
	by the proportional rule.

\end{ex_rules}

The social decision 
is the alternative that receives
the majority of overall weight.
In the case of a tie,
we assume that
each alternative is chosen
with probability $\frac{1}{2}$.
Thus,
given the rules
$\phi=(\phi_i)_{i=1}^n$
and the group-wide
vote margins
$\theta=
(\theta_i)_{i=1}^n$,
the social decision $d_\phi(\theta)$
is determined as follows:
\begin{equation}
d_\phi(\theta)
=\begin{cases}
\text{sgn }\sum_{i=1}^nw_i\phi_i(
\theta_i)
&\text{if $\sum_{i=1}^n
	w_i\phi_i(\theta_i)\neq0$}, \\
\pm1
\text{ equally likely }
&
\text{if $\sum_{i=1}^n
	w_i\phi_i(\theta_i)=0$}.
\end{cases}
\label{eq:decision}
\end{equation}

\subsection{The Game}
\label{sec:gamma}

We now define the non-cooperative game
$\Gamma$
in which the $n$ groups
choose their own rules simultaneously.

The game
is played
under incomplete
information
about
individuals' preferences,
and hence about
the group-wide
vote margins.
Each 
group chooses
a rule so as
to maximize
the expected
welfare of its members.
Since rules are fixed
prior to realization
of the preferences,
a pure strategy of the game is a function from the realization of members' preferences to the allocation of the weight.  
%since
%in most actual collective decisions, 
%the

Let $\Theta_i$
be
a random variable
that takes values
in $[-1,1]$
and represents
the vote margin
in group $i$.\footnote{Throughout the paper, we use capital $\Theta_i$ for the representation of a random variable, and small $\theta_i$ for the realization.}
We impose little
restriction
on the joint
distribution
of the
random vector
$\Theta=(\Theta_i)_{i=1}^n$.
The precise
assumption
on the 
distribution
will be stated
later in this section (Assumption
\ref{as:general}).

The ex post payoff for group $i$ is the
average payoff for its members from the social decision.
For simplicity,
we assume that each individual obtains
payoff 1 if he prefers the social decision
and payoff $-1$ otherwise.\footnote{
This assumption is introduced to facilitate the interpretation of $\theta_i$ as a vote margin. 
But even when the ex post payoffs admit heterogeneous preference intensities, the limitation imposed by the assumption is not essential, because there exists an affine transformation between payoffs with and without the assumption, rendering the strategic incentive equivalent, as we show in Remark \ref{rem:cardinal}.}
The average payoff of members of group $i$
equals $\Theta_i$ or $-\Theta_i$
depending on whether the social decision
is $+1$ or $-1$;
more concisely,
it is:
\[
\Theta_id_\phi(\Theta).
\]

The ex ante payoff 
for group $i$,
denoted $\pi_i(\phi)$,
is
the expected value of the  above expression:
\begin{equation}
\pi_i(\phi)=\mathbb{E}\left[\Theta_id_\phi(\Theta)\right].
\label{eq:payoff}
\end{equation}

Let $\pi_i(x_i,\phi_{-i}|\theta_i)$
denote the interim payoff
for group $i$
if it chooses the weight margin
$x_i\in[-1,1]$
given the realization
of the vote margin $\theta_i$.
It is obtained
as the weighted average of
the ex post payoffs $\theta_i$
and $-\theta_i$
from decisions $+1$
and $-1$ with the conditional probabilities:\footnote{The term corresponding
	to the event of a tie ({i.e.},
	$w_ix_i+\sum_{j\neq i}w_j\phi_j(\Theta_j)=0$)
	does not appear in the formula below,
	since we assume that the tie is broken
	fairly.}
\begin{equation}
\begin{split}
&\pi_i(x_i,\phi_{-i}|\theta_i)\\
&=\theta_i
\mathbb{P}\left\{w_ix_i
+\textstyle\sum_{j\neq i}w_j\phi_j(\Theta_j)>0\Big|\Theta_i=\theta_i\right\}\\
&\hspace{0.5cm}-\theta_i
\mathbb{P}\left\{w_ix_i+\textstyle\sum_{j\neq i}w_j\phi_j(\Theta_j)<0\Big|\Theta_i=\theta_i\right\}.
\end{split}
\label{eq:conditional_welfare}
\end{equation}
The ex ante
and interim payoffs
are thus related as follows:
\begin{equation*}
\begin{split}
&\text{$\phi_i$
	maximizes
	$\pi_i(\phi_i,\phi_{-i})$}\\
&\text{$\Leftrightarrow$
	$x_i=\phi_i(\theta_i)$
	maximizes
	$\pi_i(x_i,\phi_{-i}|\theta_i)$
	for almost every
	$\theta_i\in[-1,1]$.}
\end{split}
\end{equation*}

\bigskip

To summarize,
the
\textit{game $\Gamma$}
is the one in which:
the players are the $n$ groups;
the strategy set
for each group $i$
is the set $\Phi$ of all rules;
the payoff function
for group $i$ is $\pi_i$
defined in (\ref{eq:payoff}).

\bigskip

The following is the assumption on the
joint distribution of the group-wide margins.

\begin{as}
	The joint distribution
	of group-wide margins $(\Theta_i)_{i=1}^n$
	is absolutely continuous
	and has full support $[-1,1]^n$.
	\label{as:general}
\end{as}

Assumption \ref{as:general}
permits a wide variety
of joint distributions
of individuals' preferences,
in which
intra- and inter-group correlations
and biases are possible.
First, the assumption
imposes no restriction
on preference correlations within
each group.
Second, individuals' preferences may also be correlated across 
groups,
since the group-wide margins $(\Theta_i)_{i=1}^n$
can be correlated.
This allows
us to capture situations
where, for instance,
residents
of different states
or members of
different parties
have common
interest
on some issues.
Third,
preferences
may be
biased toward a particular alternative,
since
$\Theta_i$
can be asymmetrically distributed.
For instance, 
blue (resp. red) states in the US
might be described
as groups
whose group-wide margins
have a distribution
biased to the left (resp. right).
In contrast,
swing states
might be described as
groups whose distributions
are concentrated
around zero.

\begin{rem} \label{rem:latent}
	{\textit{Success probability
			and
			voting power.}}
	Our definition of group payoffs has the following interpretation based on
	the members' preferences. Let $M_{i}$ be the set of individuals in group $i$%
	, and  $X_{im}\in \left\{ -1, +1\right\} $ be the preferred alternative of
	member $m\in M_{i}$ in group $i$. 
	{Let us here redefine
		$\Theta_i$ as a \textit{latent variable}
		that parametrizes the distribution
		of the random preferences in group $i$. Specifically,}
	suppose $X_{im}$ are independently and
	identically distributed conditional on the realization $\left( \theta
	_{i}\right) _{i=1}^{n}$ with the following probabilities for all $i=1,\cdots
	,n$ and $m\in M_{i}$:%
	\begin{equation}
	\left\{ 
	\begin{array}{c}
	\mathbb{P}\left\{ X_{im}=+1|\Theta _{1}=\theta _{1},\cdots ,\Theta
	_{n}=\theta _{n}\right\} =\left( 1+\theta _{i}\right) /2, \\ 
	\mathbb{P}\left\{ X_{im}=-1|\Theta _{1}=\theta _{1},\cdots ,\Theta
	_{n}=\theta _{n}\right\} =\left( 1-\theta _{i}\right) /2.%
	\end{array}%
	\right.   \label{eq:prob_Xim}
	\end{equation}%
	{Then, as the group size
		becomes large
		($|M_i|\to\infty$),
		the Law of Large Numbers implies
		that the group-wide margin
		$\frac{1}{M_i}\sum_{m\in M_i}
		X_{im}$
		indeed converges to $\Theta_i$
		almost surely,
		which is consistent
		with our original definition of
		$\Theta_i$
		as the group-wide margin. Moreover,}
	\begin{equation*}
	\begin{split}
		&\mathbb{P}\left\{ X_{im}=d_{\phi }(\Theta)\right\}\\
		  &=\mathbb{E}\left[ \mathbb{P}%
		\left\{ X_{im}=d_{\phi }(\Theta)|\Theta \right\} \right]  \\
		&=\mathbb{E}\left[ \mathbb{P}\left\{ X_{im}=1,d_{\phi }(\Theta)=1|\Theta \right\} +%
		\mathbb{P}\left\{ X_{im}=-1,d_{\phi }(\Theta)=-1|\Theta \right\} \right]  \\
		&=\mathbb{E}\left[ \mathbb{P}\left\{ d_{\phi }(\Theta)=1|\Theta \right\} \frac{%
			1+\Theta _{i}}{2}+\mathbb{P}\left\{ d_{\phi }(\Theta)=-1|\Theta \right\} \frac{%
			1-\Theta _{i}}{2}\right]  \\
		&=\frac{1}{2}\left( 1+\mathbb{E}\left[ \mathbb{P}\left\{ d_{\phi }(\Theta)=1|\Theta
		\right\} \Theta _{i}+\mathbb{P}\left\{ d_{\phi }(\Theta)=-1|\Theta \right\} \left(
		-\Theta _{i}\right) \right] \right)  \\
		&=\frac{1}{2}\left( 1+\mathbb{E}\left[ \Theta _{i}d_{\phi }(\Theta)\right] \right) .
		\end{split}
	\end{equation*}%
	Therefore, $\pi _{i}\left( \phi \right) =\mathbb{E}\left[ \Theta _{i}d_{\phi }(\Theta)\right] $\ is an affine transformation of the probability that the
	preferred alternative of a member $m$ in group $i$ coincides with the social
	decision ($X_{im}=d_{\phi }(\Theta)$), which is called {\textit{success}} in the literature of voting power measurement (\cite{LaruelleValenciano2008}). The objective of the group, formulated as the maximization of $\pi_{i}$, is thus equivalent to maximization of the probability of success.
	
	Under the winner-take-all profile $\phi^{\rm WTA}$, $\pi _{i}$ is closely related to the classical voting power indices studied in the literature. 
	%The distribution (\ref{eq:prob_Xim}) also allows us to connect 
	If $\left( \Theta
	_{i}\right) _{i=1}^{n}$ are independently, identically and symmetrically
	distributed 
	(thus each group's preferred alternative is independently and equally distributed over $\left\{ -1,+1\right\} $, called \textit{Impartial Culture}), 
	then 
	$\pi_{i}$ corresponds to the Banzhaf-Penrose index (\cite{Banzhaf1965}, \cite{Penrose1946})
	and $\mathbb{P}\left\{X_{im}=d_{\phi }(\Theta)\right\} $ to the Rae index (\cite{Rae1969}), up to a multiplication by the constant $\mathbb{E}\left[ \left| \Theta_i \right| \right]$. 
	If $\left( \Theta _{i}\right) _{i=1}^{n}$ are perfectly correlated and
	symmetrically distributed (called \textit{Impartial Anonymous Culture}; see, for
	example, \cite{lebreton2016JMathE}), then $\mathbb{\pi }_{i}$ 
	%coincides with 
	corresponds to
	the Shapley-Shubik index (\cite{ShapleyShubik1954}).\qed
	%  times a constant $\mathbb{E}\left[ \left| \Theta_i \right| \right]$. 
	\label{rem:voting_power}
\end{rem}

\begin{rem} \label{rem:cardinal}
	\textit{
Heterogeneous preference intensities.
%			Intra-group heterogeneity.
	}
	We have
	assumed that 
	all individuals
	have the same
	preference intensities
	({i.e.},
	each individual receives a unit
	payoff whenever
	she prefers the social
	decision),
	and that 
	each
	group's
	objective
	is to maximize 
	the ex ante average
	 payoff
	of its members.
	However,
	our formal definition
	(\ref{eq:payoff})
	can be generalized to heterogeneous preference intensities.
	It
	only suffices for 
	the 
	group-wide payoff
	from the social decision
	to be more generally  defined,
	not necessarily as 
	the average
	of members' payoffs with identical preference intensities.

	To be more
	precise,
	suppose each group $i$
	receives a random
	payoff
	$U_i^+$
	or $U_i^-$
	depending on
	whether the
	social decision
	is
	$+1$ or $-1$,
	where $U_i^+$ and $U_i^-$
	are assumed to
	take \textit{any} values in $[0,1]$.
	Redefine
	the variable
	$\Theta_i$
	as
	the payoff difference:
	$\Theta_i:=
	U_i^+-
	U_i^-$.
	Then 
	the group's
	ex ante
	payoff
	from the social decision
	under profile $\phi$
	is
	\begin{equation*}
	\begin{split}
	u_i(\phi)&=
	\mathbb{E}
	\left[
	U_i^+
	\frac{1+d_\phi(\Theta)}{2}
	+
	U_i^-
	\frac{1-d_\phi(\Theta)}{2}
	\right]\\
	&=\frac{1}{2}\mathbb{E}
	\left[
	\Theta_id_\phi(\Theta)
	\right]
	+\frac{1}{2}
	\mathbb{E}\left[
	U_i^++U_i^-
	\right]\\
	&=\frac{1}{2}
	\pi_i(\phi)+
	\text{constant}.
	\end{split}
	\end{equation*}
	Since this
	is 
	a positive
	affine transformation
	of 
	$\pi_i(\phi)$,
	our model
	captures the general case
	where each group
	maximizes the expected
	group-wide payoff
	$u_i$.
	In particular,
	the group-wide payoffs
	$U_i^+$ and $U_i^-$ can
	be
	any functions
	of members' payoffs
%	which may or may not respect different
including heterogeneous
	preference intensities.

Furthermore, by considering groups in our model as agents whose preference intensities are private information, we can consider our model in the context of more general and abstract Bayesian mechanism design problems.
The $n$-agent setting will be useful in Section \ref{sec:imp_mech} 
where we clarify the underlying logic behind the results we obtain in Section \ref{subsec:result}. 
For that purpose, we introduce a formal definition of the social choice function in the following Subsection \ref{subsec:scf}.
\qed
\end{rem}

\subsection{Social Choice Functions}
\label{subsec:scf}

%Although the social decision
%in game $\Gamma$
%is determined by the sum
%of the groups' weights allocated to
%each alternative,
%it would be useful to consider more general social choice
%functions in order to understand the welfare properties obtained for game $\Gamma$
%in the context of $n$-group mechanism design problems.

Let $\triangle\left( \{-1,+1\} \right)$ be the set of all random variables taking values in $\{-1,+1\}$.
A \textit{social choice function (SCF)}
is a Borel-measurable function\footnote{More precisely, an SCF is a function
	$d(\theta,\omega)$
	of two variables, $\theta\in[-1,1]^n$ and
	$\omega\in\Omega$
	for a sample space $\Omega$,
	such that: for each
	$\theta$,
	$d(\theta,\cdot):\Omega\to\{-1,+1\}$
	is a random variable;
	for each $\omega$,
	$d(\cdot,\omega):[-1,1]^n\to\{-1,+1\}$ is
	a Borel-measurable function.}
\[
d:[-1,1]^n \to \triangle\left( \{-1,+1\} \right) .
\]
The SCF assigns to
each profile of realized vote margins
$\theta=(\theta_i)_{i=1}^n\in[-1,1]^n$
a social decision $d(\theta)$
which may randomize between alternatives $-1$ and $+1$.
The  decision function $d_\phi$
in game $\Gamma$
is 
an example of an SCF.\footnote{Randomness
of $d_\phi(\theta)$ occurs
when the weighted vote is tied.}

With a slight abuse of notation,
we denote by $\pi_i(d)$
the ex ante payoff for group $i$
under SCF $d$.
By extending formula (\ref{eq:payoff}),
we have the following expression:
\[
\pi_i(d)=\mathbb{E}\left[\Theta_id(\Theta)\right].
\]

Our main analysis in Section \ref{subsec:result}
is based on game $\Gamma$,
but
the results have implications
to mechanism design problems with general SCFs,
which we summarize 
in Section \ref{sec:imp_mech}.

\section{The Dilemma}

\subsection{The Main Result}
\label{subsec:result}

In game $\Gamma$,
a rule
(or strategy) $\phi_i$
for group $i$
\textit{{weakly}
	dominates}
another rule $\psi_i$
if
$\pi_i(\phi_i,\phi_{-i})\geq\pi_i(\psi_i,\phi_{-i})$
for any $\phi_{-i}$,
with strict inequality
for at least one
$\phi_{-i}$.
A rule $\phi_i$ is a
\textit{{weakly} dominant
	strategy}
for group $i$
if it weakly dominates
every rule not equivalent
to $\phi_i$
where we call two rules $\phi_i$
and $\psi_i$
\textit{equivalent}
if $\phi_i(\theta_i)=\psi_i(\theta_i)$
for almost every $\theta_i$
(with respect to Lebesgue measure
on $[-1,1]$).

A profile $\phi$
\textit{Pareto dominates}
another profile $\psi$
if $\pi_i(\phi)\geq\pi_i(\psi)$
for all $i$,
with strict inequality
for at least one $i$.
If $\phi$
is not Pareto dominated
by any profile,
it is called
\textit{Pareto efficient}.
Pareto dominance
between SCFs is
defined
in the same way,
based on the payoff functions
$\pi_i(d)$ (see Section \ref{subsec:scf}).

We first consider the case in which there is no `dictator' group that can determine the winner by putting all its weight to one alternative (Theorem \ref{thm:dilemma}). 
Later we consider the case with such a group (Proposition \ref{prop:dict}).

\begin{as}
	Each group
	has less than half the total weight:
	$w_i<\frac{1}{2}\sum_{j=1}^n w_j$ for all $i=1,\cdots,n$.
	\label{as:no_dictator}
\end{as}

\begin{theorem}
	Under Assumptions
	\ref{as:general}
	and \ref{as:no_dictator},
	game $\Gamma$
	is a Prisoner's Dilemma: 
	\begin{itemize}
		\item[(i)]
		the winner-take-all rule $\phi_i^{\rm WTA}$
		is the weakly
		dominant strategy\footnote{By the definition of weak dominance,
		$\phi_i^{\rm WTA}$
		is the unique
		weakly dominant
	strategy up to equivalence of rules.}
		for each group $i$;
		\item[(ii)]
		the winner-take-all profile
		$\phi^{\rm WTA}
		$
		is Pareto dominated.
	\end{itemize}
	\label{thm:dilemma}
\end{theorem}

We use the following lemma
to prove the theorem.
{An SCF $d$ is called
a \textit{weighted majority rule}
if there exists a vector $(\lambda_i)_{i=1}^n
\in\mathbb{R}^n_+\setminus\{\textbf{0}\}$
such that:}
\[
d(\theta)=\text{ sgn }\sum_{i=1}^n\lambda_i\theta_i
\text{ for almost every $\theta\in[-1,1]^n$.}
\]
In game $\Gamma$,
a profile $\phi$
is called a \textit{generalized proportional
profile}
if there exists a vector
$(\lambda_i)_{i=1}^n\in
[0,1]^n\setminus\{\textbf{0}\}$
such that for each $i$,
\[
\phi_i(\theta_i)=\lambda_i\theta_i
\text{ for almost every $\theta_i\in[-1,1]$.}
\]
Two profiles $\phi$
and $\psi$
are called
\textit{equivalent}
if $d_\phi(\theta)=d_\psi(\theta)$
for almost every $\theta\in[-1,1]^n$.

\begin{lem}
	(Characterization of the Pareto
	set)
	Under Assumption
	\ref{as:general},
	the following statements hold:
	\begin{itemize}
		\item[(i)]
		An SCF $d$ is Pareto efficient
		in the set of all SCFs
		if and only
		if it is a weighted majority
		rule.
		\item[(ii)]
	In game $\Gamma$, a profile
	$\phi=(\phi_i)_{i=1}^n$ is Pareto
	efficient in the set of all profiles
	if and only if
	it is equivalent to a generalized proportional profile.
	\end{itemize}
	\label{lem:char}
\end{lem}

The proof of Lemma \ref{lem:char} is relegated to the Appendix.

\begin{proof}[Proof of Theorem \ref{thm:dilemma}]
	\textit{Part (i).}
	We first check
	that
	\begin{equation}
	\pi_i(\phi_i^{\rm WTA},\phi_{-i})
	\geq
	\pi_i(\phi_i,\phi_{-i})
	\label{eq:weak_ineq}
	\end{equation}
	for any $(\phi_i,\phi_{-i})$.
	By (\ref{eq:conditional_welfare}),
	if $\theta_i>0$ (resp. $\theta_i<0$),
	then the interim payoff
	$\pi_i(x_i,\phi_{-i}|\theta_i)$
	is non-decreasing (resp. non-increasing)
	in $x_i\in[-1,1]$.
	We thus have
	$\pi_i(\phi_i^{\rm WTA}(\theta_i),\phi_{-i}|\theta_i)
	\geq
	\pi_i(\phi_i(\theta_i),\phi_{-i}|\theta_i)$
	for any $(\phi_i,\phi_{-i})$
	and $\theta_i\neq0$.
	Since $\Theta_i=0$ occurs
	with probability 0,
	This implies (\ref{eq:weak_ineq}).

	Now we show that
	for any  profile $\phi_{-i}$
	in which each $\phi_j(\Theta_i)$
	($j\neq i$)
	has full support $[-1,1]$ (e.g.,
	$\phi_j^{\rm PR}$),
	the strict inequality
	\begin{equation}
	\pi_i(\phi_i^{\rm WTA},\phi_{-i})
	>
	\pi_i(\phi_i,\phi_{-i})
	\label{eq:strict_ineq}
	\end{equation}
	holds
	for any rule $\phi_i$
	that differs from $\phi_i^{\rm WTA}$
	on a set $A\subset[-1,1]$ of positive measure.
	To see this,
	note that for such $\phi_{-i}$
	and any
	$\theta_i$,
	the conditional distribution
	of $\sum_{j\neq i}w_j\phi_j(\Theta_j)$ given $\Theta_i=\theta_i$
	has support 
	\[
	I=\left[-\sum_{j\neq i}w_j,\sum_{j\neq i}w_j\right].
	\]
	Since 
	$w_i<\sum_{j\neq i}w_j$ by Assumption
	\ref{as:no_dictator},
	as 
	$x_i$
	moves in $[-1,1]$,
	$w_ix_i$
	moves in interval $I$.
	Formula
	(\ref{eq:conditional_welfare}) thus
	implies that
	if $\theta_i>0$ (resp. $\theta_i<0$),
	then
	$\pi_i(x_i,\phi_{-i}|\theta_i)$
	is strictly increasing (resp. decreasing)
	in $x_i\in[-1,1]$.
	Hence
	$\pi_i(\phi_i^{\rm WTA}(\theta_i),\phi_{-i}|\theta_i)>\pi_i(\phi_i(\theta_i),\phi_{-i}|\theta_i)$
	at any $\theta_i\in A$.
	Since $\Theta_i$
	has full support, result
	(\ref{eq:strict_ineq})
	follows.

	\textit{Part (ii)}.
	By the characterization
	of the Pareto set
	(Lemma \ref{lem:char}(ii)),
	it suffices to check
	that 
	$\phi^{\rm WTA}$
	is not equivalent
	to any generalized proportional profile.
	Suppose, on the contrary,
	that $\phi^{\rm WTA}$
	is equivalent
	to a generalized proportional profile 
	with coefficients
	$\lambda\in[0,1]^n\setminus\{\textbf{0}\}$.
	Then, since $(\Theta_i)_{i=1}^n$
	has full support,
	\begin{equation}
	d_{\phi^{\rm WTA}}(\theta)=\text{sgn\,}\sum_{i=1}^nw_i
	\lambda_i\theta_i\text{ at almost every }\theta\in[-1,1]^n.
	\label{eq:equiv}
	\end{equation}
	Since no group dictates
	the social decision,
	the coefficients
	$\lambda_i$ are positive
	for at least two groups.
	Without loss of generality,
	assume $\lambda_1>0$
	and $\lambda_2>0$.
	Now,
	fix $\theta_i$
	for $i\neq 1,2$
	so that they are sufficiently
	small in absolute value.
	Then, according
	to (\ref{eq:equiv}),
	for (almost any) sufficiently small $\varepsilon>0$,
	$d_{\phi^{\rm WTA}}(\theta)=+1$
	if $\theta_1=1-\varepsilon$
	and $\theta_2=-\varepsilon$,
	while
	$d_{\phi^{\rm WTA}}(\theta)=-1$
	if $\theta_1=\varepsilon$
	and $\theta_2=-1+\varepsilon$.
	This contradicts the fact
	that $d_{\phi^{\rm WTA}}(\theta)$
	depends only on
	the signs
	of $(\theta_i)_{i=1}^n$.	
\end{proof}

%{
{
Theorem \ref{thm:dilemma}
shows that,
while 
the dominant strategy for each group
is the winner-take-all rule,
the dominant-strategy equilibrium
is Pareto dominated by a generalized
proportional profile.
This typical 
Social
Dilemma (or, $n$-player Prisoner's Dilemma) situation suggests that a Pareto efficient outcome is not expected to be achieved under decentralized decision making, and a coordination device is necessary in order to attain a Pareto improvement.}

The observation that groups
have an incentive
to use the winner-take-all rule
is not new.
\citet{BeisbartBovens2008}
consider Colorado's deviation from the winner-take-all rule to the proportional rule, following the state's attempt in 2004 to amend the state constitution, 
and show that the citizens in Colorado are worse off under both \textit{a priori} %measure with the impartial culture 
and \textit{a posteriori} measures. %based on the parameter distribution calibrated by the recent election results. 
\citet{Hummel2011} shows that a majority of the voters in a state is worse off by unilaterally switching to the proportional rule from the winner-take-all profile. 

Our results are also consistent with the findings in the literature of the coalition formation games in which 
individuals may have incentive to raise their voices by forming a coalition and aligning their votes.
%individuals' incentive to align their votes is studied using the coalition formation games.
\citet{Gelman2003} illustrates that individuals are better off by forming a coalition and assign all their weights to one alternative.
%, even though the average payoff of all individuals may decrease. 
\citet{Eguia2011GEB} considers a game in which the members in an assembly decide whether to accept the party discipline to align their votes, 
and shows that the voting blocs form in equilibrium if preferences are sufficiently polarized. 
\citet{Eguia2011AJPS} considers a dynamic model and shows the conditions under which voters form two polarized voting blocs in a stationary equilibrium. 
%In that sense, our results are in line with their findings.

A novelty of Theorem \ref{thm:dilemma}
lies in its generality. 
Earlier studies have
introduced a specific structure
either 
on the distribution of the preferences and/or of the weights, 
or on the set of the rules that groups can use.\footnote{ 
\citet{BeisbartBovens2008} consider Colorado's strategic choice between the winner-take-all and the proportional rules in the US Electoral College.
\citet{Hummel2011} either introduces a correlation structure in the preference distribution or assumes weights to be constant in other states. 
\citet{Gelman2003} shows interesting computations, but all claims are based on observations from examples.
\citet{Eguia2011GEB} introduces a three-group preference structure, left, right and independent, and 
\citet{Eguia2011AJPS}'s main results focus on a nine-voter example, and the internal rules are assumed to be (super) majority rules.
} 
In contrast, we only impose 
fairly mild conditions on the preference distribution
(in particular, Assumption \ref{as:general} imposes no restriction on across-group correlation),
on the weight distribution (Assumption \ref{as:no_dictator} imposes no specific weight structure 
such as one big group and several smaller ones, or equally sized groups), and  
on the set of the available rules (Definition \ref{def:rule} admits all Borel-measurable rules, not just the winner-take-all and the proportional rules). 

Most importantly, the generality of our model allows for a welfare analysis 
which does not require introduction of a
specific structure on 
the weight and/or
the preference distribution, or the set of available rules.
Since our model incorporates \textit{all} Borel-measurable profiles, 
%largely beyond the comparison of the winner-take-all and the proportional rules, 
the Pareto set obtained in Lemma \ref{lem:char}
leads us to an explicit characterization 
of 
the set of \textit{first-best} outcomes
which can be attained. 

%An immediate welfare implication of the characterization lemma is that the dominant-strategy equilibrium is Pareto dominated, because the winner-take-all rule is not equivalent to a generalized proportional rule. 
%In addition, the lemma tells us what first-best profiles Pareto dominate the equilibrium.
%This is useful in welfare analysis, because once we know the profiles, it tells us what the society can aim for,  if a coordination device to escape from the Prisoner's Dilemma is available.
%by the generalized proportional profiles. 
%It implies that the game is a Prisoner's Dilemma and that if the society wants to escape from the dilemma, it should use a coordination device which aims to achieve the outcomes induced by a generalized proportional profile.
The key welfare implication of our result is that the dominant-strategy equilibrium is Pareto dominated by generalized proportional profiles. This provides us with two important insights in welfare analysis of groupwise preference aggregation problems. First, the game is a Prisoner's dilemma so that a coordination device is necessary for a Pareto improvement. Second, once such a device is available, our characterization lemma tells us that, at the first-best,
the society should use 
rules that are proportional in nature,
so that the cardinal information of the group-wide preferences is transmitted without distortion.

It is worth emphasizing that the result does not imply merely utilitarian ({i.e.}, benthamite) inefficiency of the equilibrium profile. 
The profile is Pareto dominated, implying that it is in \textit{every}
group's interest
to move from
the winner-take-all equilibrium
to another profile.
From the utilitarian perspective,
it is straightforward
to see that the social optimum
is obtained by the
\textit{popular vote},
{i.e.},
direct majority voting
by all individuals.
%which is equivalent to a generalized proportional profile with specific coefficients.
However,
this observation
is not sufficient
to establish that the winner-take-all profile
is Pareto dominated.\footnote{Obviously, utilitarian optimality does not imply Pareto dominance. 
}
%Utilitarian consideration is silent about whether other profiles are Pareto efficient or not.
After all, the utilitarian optimum is merely one point in the Pareto set.

The following example illustrates
that 
the winner-take-all profile
is not always Pareto dominated by 
either the popular vote or the proportional profile.

%\subsection{Implications of the Characterization of the Pareto Set}
%\label{sec:imp_char}

\begin{ex}
\label{ex:3groups}	
	Consider a 
	society which consists of
    two large groups 
    with an equal weight and one small group.
	For an illustrative purpose, 
	let us consider three American states: Florida, New York and Wyoming.
	Their populations and weights
	are summarized in Table \ref{tab:FLNYWY}.

\begin{table}[h]
	\caption{Comparison of the expected payoffs in an example of the society which consists of three states: Florida, New York and Wyoming. Weights are the electoral votes assigned in the Electoral College in 2020. Population is an estimation of the voting-age population in 2018 (in thousands). Source: US Census Bureau. }
	\centering
		\resizebox{\columnwidth}{!} {
	\begin{tabular}{lcccccc} \hline\hline
		State & Weight & Population & $\pi_i\left(\phi^{\rm WTA}\right)$ 
								    & $\pi_i\left(\phi^{\rm PR}\right)$ 
								    & $\pi_i\left(\phi^{\rm POP}\right)$ 
								    & $\pi_i\left(\hat{\phi}\right)$ \\ 
		\hline
		Florida   & 29 & 15,047 & 0.250 & 0.332 & 0.343 & 0.271\\
		New York  & 29 & 13,684 & 0.250 & 0.332 & 0.323 & 0.271\\
		Wyoming   & 3  & 422   & 0.250 & 0.034 & 0.008  & 0.271
		\\ \hline
		Per capita average & & & 0.250 & 0.328 & 0.329 & 0.271 \\ \hline
	\end{tabular}
	}
	\label{tab:FLNYWY}
\end{table}
	
	As defined above, $\phi^{\rm WTA}$ and $\phi^{\rm PR}$
	are the winner-take-all and proportional
	profiles.
	The vote margins $\left(\Theta_i \right)_{i=1,2,3}$ are drawn from the uniform distribution on $[-1, 1]$ independently across the states.
	The payoff of the popular vote $\pi_i\left(\phi^{\rm POP}\right)$ is defined as 
	the ex ante expected payoff of a representative voter in each state, which is obtained by letting the social decision $d$ be the popular vote winner in (\ref{eq:payoff}).

	Since there is no dictator state ({i.e.}, Assumpion \ref{as:no_dictator} is satisfied) in this example, any pair of two states is a minimal winning coalition under the winner-take-all profile, implying that the expected payoffs are exactly the same across states under $\phi^{\rm WTA}$.
	
	The two larger states are better off under the proportional profile $\phi^{\rm PR}$, while the smaller state is worse off. 
	This is because the social decision is more likely to coincide with
	the alternative preferred by the majority of the large states under $\phi^{\rm PR}$.
	As a consequence, the differences in the weights are reflected more directly on the differences in the expected payoffs. 
	
	Even though the small state is better off under $\phi^{\rm WTA}$ in this particular example, it is worth underlining that whether the winner-take-all profile favors small states as compared to the proportional profile depends on the weight distribution. 
	For example, if one state is a dictator ({i.e.} violating Assumption \ref{as:no_dictator}), the payoffs of the two other states are zero under $\phi^{\rm WTA}$, while they are (probably small but) strictly positive under $\phi^{\rm PR}$.

	Under the popular vote $\phi^{\rm POP}$, the expected payoff of the small state is even smaller than under $\phi^{\rm PR}$. This comes from the fact that the weight assigned to the small state is larger than the large states in the \textit{per capita} measure. In this example, Wyoming has more weight than it would if assigned proportionally to the population.\footnote{The digressive proportionality is a consequence of the rule specified in the US Constitution. The number of electoral votes of each state is the sum of the numbers of Senate members (constant) and of the House (proportional to population in principle). Under such a rule, per capita weight is decreasing in population.} 
	Under the popular vote, the citizens in the small state lose such an advantage assigned through the weights. 
	We can also observe that the utilitarian (benthamite) welfare is maximized under the popular vote $\phi^{\rm POP}$ by comparing the per capita average of the expected payoffs.

%	These numbers represent ... in the US.
%	Popular
%	vote $\phi^{\rm PV}$
%	is defined as the
%	proportional profile under 
%	the weight assignment
%	$(w_1,w_2,w_3)=(n_1,n_2,n_3)$,
%	or equivalently,
%	the generalized proportional 
%	profile with coefficients
%	$\lambda_i=\frac{n_i}{w_i}$.
	Finally, let $\hat{\phi}$
	be the generalized proportional
	profile with coefficients
	$\lambda_i=1/w_i$. We observe that it 
	Pareto dominates $\phi^{\rm WTA}$.
	Remember that our characterization lemma tells us that a profile is Pareto efficient if and only if it is equivalent to a generalized proportional profile. We can show that among the profiles which Pareto dominate the equilibrium profile $\phi^{\rm WTA}$, one is obtained by letting $\lambda_i=1/w_i$, because the expected payoffs are equal across the states in this example, and we can obtain the particular point in the Pareto set with the equal Pareto coefficients by setting $\lambda_i=1/w_i$. 
	
	This example illustrates that the winner-take-all, proportional profiles, and the popular vote may be all Pareto imcomparable.
	Even though Theorem \ref{thm:dilemma} shows that the winner-take-all profile is Pareto dominated, it may not be dominated by either the proportional profile or the popular vote. 
	This may happen when the number of groups is small. 
	For the cases in which there are sufficiently many groups,
	we provide clear-cut insights in Section \ref{sec:asym} by using an asymptotic model and numerical simulations.
	\qed
\end{ex}

To summarize, we have the following propositions.

\begin{prop}
	Under Assumption \ref{as:general}, the proportional profile $\phi^{\rm PR}$ and the popular vote $\phi^{\rm POP}$ are both Pareto efficient.
\end{prop}
\begin{proof}
	Trivially, the proportional profile is a generalized proportional profile by letting $\lambda_i =1$ for all $i$. The outcome of the popular vote coincides with that of the generalized proportional profile with $\lambda_i = n_i/w_i$ for all $i$.
	By Lemma \ref{lem:char}, we obtain the result.
\end{proof}

\begin{prop}
	\label{prop:dict}
	Under Assumption \ref{as:general}, the winner-take-all profile $\phi^{\rm WTA}$ is Pareto dominated if and only if Assumption \ref{as:no_dictator} is satisfied.
\end{prop}
\begin{proof}
	The ``if'' part is already proven in Theorem \ref{thm:dilemma} (ii). To show the ``only if'' part, suppose that Assumption \ref{as:no_dictator} is violated. Then, there exists a dictator state $i^\ast$ that can determine the winner by putting all its weight to the alternative preferred by the majority of the state.  
	Hence, $\phi^{\rm WTA}$
	is equivalent to the generalized
	proportional profile with coefficients
	$\lambda_{i^\ast}>0$
	and $\lambda_i=0$
	for all $i\neq i^\ast$. 
	By Lemma \ref{lem:char}, we obtain the result.	
\end{proof}

\subsection{An impossibility theorem underlying the WTA Dilemma}
\label{sec:imp_mech}
In order to provide an interpretation of the result obtained in Theorem \ref{thm:dilemma} in the context of mechanism design, we state an impossibility theorem that underlies the winner-take-all dilemma. 

We show in Remark \ref{rem:cardinal} above that there is a direct analogy between the non-cooperative voting game $\Gamma$ considered in Theorem \ref{thm:dilemma} and the Bayesian collective decision problem in which  
each agent's preferences including the intensity level are private information. 
In order to elucidate the logic behind our theorem, it is thus useful to consider a model of social choice function of which the cardinal preferences are the input.

Consider a society
which consists of
 $n$
agents ($i=1,\cdots,n$)
and which makes a collective decision
between two alternatives $+1$
and $-1$.
{As we described
	in Remark \ref{rem:cardinal},
	each group is a player in the voting game $\Gamma$, while it can be seen more generally as an agent whose preference intensity is represented by a von Neumann-Morgenstern
	utility function.}
Let $U_i^+$ (resp. $U_i^-$) be
agent $i$'s utility from the alternative
$+1$ (resp. $-1$).
We assume that
the utilities
are random variables whose support is included in a bounded interval, which we suppose as $[0,1]$
without loss of generality.

Agent $i$'s \textit{type} is represented by 
the utility difference
$\Theta_i:=U_i^+-U_i^-$.
Then, $\Theta_i$
is a random variable 
taking a value in $[-1,1]$.
The type is private information:
each agent observes only his own type.
We only impose absolute continuity and full support of
the joint distribution
(Assumption \ref{as:general}).
This allows for correlations of types,
and ex ante asymmetries with respect to the agents and the alternatives.

A social choice function (SCF) $d$
is defined as in Section \ref{subsec:scf}.
For each profile of realized types
$\theta=(\theta_i)_{i=1}^n$, 
$d(\theta)$ is a 
random variable which takes a value either $+1$ or $-1$.
An SCF 
is \textit{dictatorial}
if there
exists an agent
$i$ such that $d(\theta)$ assigns probability one to the alternative $\text{sgn }\theta_i$
for almost every $\theta\in[-1,1]^n$.
Note that 
a weighted majority rule
is dictatorial if and only if 
$\lambda_i>0$
for one $i$
and $\lambda_j=0$ for all $j\neq i$.

We consider the \textit{direct mechanism}
associated with SCF $d$.
Each of $n$ agents simultaneously
reports a type,
based on which an alternative
is chosen according to $d$.
A strategy for
agent $i$ is
a Borel-measurable function
$\sigma_i:[-1,1]\to[-1,1]$
that assigns to each realization of type
$\theta_i\in[-1,1]$ a reported type $\sigma_i(\theta_i)\in[-1,1]$.
A strategy $\sigma_i$
is called \textit{truthful}
if
$\sigma_i(\theta_i)=\theta_i$
for almost every $\theta_i$.
Given
a strategy profile $\sigma=(\sigma_i)_{i=1}^n$,
the ex ante payoff for agent $i$ induced by $d$
is:
\[
\pi_i(\sigma;d)=\mathbb{E}\left[\Theta_i
d(\sigma(\Theta))\right]
\label{eq:pi_d}
\]
where $\sigma(\Theta)=\left(
\sigma_j(\Theta_j)
\right)_{j=1}^n$
is the profile of reported types.

The game $\Gamma$ defined
in Section \ref{sec:gamma}
is thus exactly the one induced by the direct mechanism 
associated with the weighted majority
rule $d$ with coefficients $\lambda_i=w_i$
($i=1,\cdots,n$).
Call group $i$ in game $\Gamma$
as agent $i$,
and its group-wide vote margin $\Theta_i$
as the agent's type.
The strategy set $\Phi$ for group $i$ in that game
is the same as the strategy set
for agent $i$ in the direct mechanism.
The definition (\ref{eq:decision}) of the social decision $d_\phi(\theta)$
in game $\Gamma$
is exactly the same as the decision $d(\phi(\theta))$
in the direct mechanism
in which the strategy profile 
$\sigma$ coincides with $\phi$.
Therefore, the ex ante payoff functions
in the two models also coincide.

An SCF is \textit{Bayesian incentive compatible} (BIC)
if the profile
of truthful strategies
is a Bayesian Nash equilibrium
of the direct mechanism.
%\footnote{The profile
%	$\sigma^\ast$ of truthful strategies is
%	a \textit{Bayesian Nash equilibrium}
%	if for each agent $i$,
%	$\pi_i(\sigma^{\ast};d)
%	=\max_{\sigma_i}\pi_i(\sigma_i,\sigma_{-i}^\ast;d)$.}
By the revelation principle,
it is without loss of generality
to consider only direct mechanisms.

The following is the impossibility result which underlies the WTA dilemma.

\begin{prop}
	Under
	Assumption \ref{as:general},
	an SCF
		is Pareto
		efficient and
		Bayesian incentive compatible
		if and only if it is dictatorial.
	\label{prop:mech}
\end{prop}

\begin{proof}
	It is obvious that every dictatorial
	SCF is Pareto efficient and Bayesian incentive compatible.
	By Lemma \ref{lem:char}(i),
	it suffices to check that 
	if a weighted majority rule $d$ is not dictatorial, then it is not Bayesian incentive
	compatible.
	In the proof
	of Theorem
	\ref{thm:dilemma}(i),
	we have shown
	that
	in game $\Gamma$,
	if $\phi_{-i}$
	is such that
	each $\phi_j(\Theta_j)$ ($j\neq i$) has full support $[-1,1]$,
	the unique
	(up to equivalence)
	best response for group $i$
	is the winner-take-all rule.
	Thus,
	in the direct mechanism for $d$,
	the unique (up to equivalence) best response for
	each agent $i$ against
	the profile in which all other agents play a truthful strategy is $\sigma_i(\theta_i)
	=\text{ sgn }\theta_i$, which is again not
	a truthful strategy.
	Thus the profile of truthful strategies
	is not a Bayesian Nash equilibrium.
\end{proof}

The essence of the impossibility described in Proposition \ref{prop:mech} lies in 
the fundamental incompatibility between 
Pareto efficiency and equilibrium behavior in the cardinal preference aggregation problem.

In order to understand where the incompatibility comes from, consider the classical Gibbard-Satterthwaite Theorem, which states impossibility of achieving both strategyproofness and non-dictatorship in the \textit{ordinal} preference aggregation problem.
Relaxing the strategyproofness condition 
to 
Bayesian incentive compatibility
requires the introduction of expected payoff, 
as BIC is defined on the solution concept of Bayesian Nash equilibrium. 
This means that we need to consider a \textit{cardinal} preference aggregation problem.

%When cardinal payoffs attached to the alternatives are random variables,
%not only the ordinal preferences but also the cardinal intensities are private information.  
%An SCF is then defined as a function which takes information including the intensities as its input. 
%The model defined in Section \ref{subsec:scf} above falls in this class.

An impossibility result analogous to the Gibbard-Satterthwaite theorem is no longer obtained,
when the SCF takes cardinal preferences as its input. 
This is shown by a counterexample: the winner-take-all profile is 
a non-dictatorial SCF which satisfies Bayesian incentive compatibility. 
BIC alone is not sufficient to imply dictatorship.

Essentially, only ordinal information can be aggregated when Bayesian incentive compatibility is required. 
To see why, suppose that two preferences types $u$ and $v$ are in affine transformation, that is, there exists $\alpha>0$ and $\beta \in \mathbb{R}$ such that $u= \alpha v+\beta \mathbbm{1}$ (call such a transformation as \textit{purely cardinal}).
If the outcome differs by reporting between $u$ and $v$, the incentive compatibility of either $u$ or $v$ should be violated. 
To be more precise, consider a purely cardinal change in preferences.
The agent's preferences on the lotteries over the alternatives are unchanged. 
If the lottery over the alternatives changes by a purely cardinal change of the agent's report, it means that she can manipulate the outcome even though her  preferences over the lotteries are unchanged, implying a violation of incentive compatibility. 
Therefore, by requiring Bayesian incentive compatibility, the outcome should be equivalent up to purely cardinal changes,
and thus only ordinal information can be aggregated at most.\footnote{More precisely, what is invariant is the interim expected payoff, and thus the ex post collective choice may depend on cardinal intensities. However, such dependence is eventually inconsequential in the sense that the interim payoffs (and thus the ex ante payoffs) are invariant. For further discussion, see \cite{SCHMITZTROGER2012}.}   

Our observation that BIC implies ordinal aggregation is coherent with the results obtained in the literature. 
\cite{AzrieliKim2014} characterize the second-best social choice functions and show that they are ordinal qualified weighted majority rules. 
\cite{EHLERS202031}
provide a thorough analysis of the conditions under which BIC implies ordinality. 
 
On the other hand, Pareto efficiency requires aggregation of cardinal information. 
To see why, remember that a Pareto efficient allocation should solve the maximization problem of the social welfare weighted by Pareto coefficients.
By definition, weighted social welfare depends continuously on the cardinal preferences of each agent.\footnote{
Moreover, it is due to its linearity that we could provide a full characterization of the Pareto frontier by the weighted majority rules in Lemma \ref{lem:char}.} 
Therefore, cardinal information concerning the agents' preferences (such as intensity) should be reflected continuously to the social outcome when Pareto efficiency is required.

In sum, requiring incentive compatibility implies ordinal aggregation, while requiring Pareto efficiency implies cardinal aggregation. The fundamental property behind Proposition \ref{prop:mech} is the incompatibility between the two types of aggregation. 
The result is coherent with the impossibility theorem obtained in 
\cite{BORGERS20092057}
in case of three alternatives and two agents. 
\cite{EHLERS202031}
provides a general result for any number of alternatives by showing that a weighted utilitarian SCF is dictatorial if and only if it satisfies BIC under an independence condition (Theorem 8).

%When applied to the indirect voting problem,
%Proposition \ref{prop:mech} implies that 
%Pareto efficiency is not achieved under the decentralized rule-setting.
%%there is a fundamental incompatibility between decentralized rule-setting and efficiency.
%An inevitable efficiency loss arises when each group is allowed to specify a local aggregation rule independently.
%%By considering an abstract social choice function and using expressions of mechanism design, the origin of the incompatibility can be described explicitly. 
%%Proposition \ref{prop:mech} reveals the dilemma structure between groups' private incentives and their welfare consequences.
%Behind the WTA dilemma lies
%a structural incompatibility between private incentives and social welfare described by Prop 3.

When applied to the group-wide voting problem, Proposition \ref{prop:mech} provides an interpretation of our main result stated in Theorem \ref{thm:dilemma}.
Behind the dilemma structure stated in Theorem \ref{thm:dilemma}
lies the impossibility of reconciling both Bayesian incentive compatibility and Pareto efficiency, proven in a more abstract mechanism design context in Proposition \ref{prop:mech}.

\section{Asymptotic and Computational Results}
\label{sec:asym}

\subsection{Asymptotic Analysis}

We saw above that the game is a Prisoner's Dilemma. In this section, we provide further insights on the welfare properties, by focusing on the following situations in which: (i) the number of groups is sufficiently large, and (ii) the preferences of the members are distributed symmetrically. 
These properties allow us to provide an asymptotic and normative analysis.

Often the difficulty of analysis arises from the discrete nature of the problem. Since the social decision $D_\phi$ is determined as a function of the sum of the weights allocated to the alternatives across the groups, computing the expected payoffs may require the classification of a large number of success configurations which increases exponentially as the number of groups increases, rendering the analysis prohibitively costly. We overcome this difficulty by studying asymptotic properties. In order to check the sensibility of our analysis, we provide Monte Carlo simulation results later in the section, using an example of the US Electoral College.

\bigskip

In order to study asymptotic properties, let us consider a sequence of weights $(w_i)_{i=1}^\infty$, exogenously given as a fixed parameter.

\begin{as}
	The sequence of weights $(w_i)_{i=1}^\infty$ satisfies the following properties.
	\begin{itemize}
		\item[(i)]
		$w_1,w_2,\cdots$
		are in a finite
		interval $[\underline{w},\bar{w}]$
		for some $0\leq
		\underline{w}<\bar{w}$.
		\item[(ii)] 
		As $n\to\infty$, the statistical distribution $G_n$ induced by $(w_i)_{i=1}^n$
		weakly converges to a distribution $G$ with support $[\underline{w},\bar{w}]$.\footnote{The
			statistical distribution function $G_n$ induced by $(w_i)_{i=1}^n$ is
			defined by $G_n(x)=\#\{i\leq n | w_i\leq x\}/n$ for each $x$.
			$G_n$ weakly converges to $G$ if $G_n(x)\to G(x)$ at every
			point $x$ of continuity of $G$.}
	\end{itemize}
	\label{as:w}
\end{as}

Assumption \ref{as:w}
guarantees that
for large $n$, the statistical distribution
of weights
$G_n$ is sufficiently close to some
well-behaved distribution $G$,
on which our asymptotic analysis is based.

Additionally, we impose an impartiality assumption for our normative analysis: 
\begin{as}
	The variables $(\Theta_i)_{i=1}^\infty$
	%defined in Assumption \ref{as:general} 
	are drawn independently
	from a common symmetric distribution $F$.
	\label{as:cor}
\end{as}

As in \cite{FelsenthalMachover1998}, a normative analysis requires impartiality, and a study of fundamental rules in the society, such as a constitution, should be free from specific dependence on the ex post realization of the group characteristics. 
Assumption \ref{as:cor} allows our normative analysis to abstract away the distributional details. Of course, a normative analysis is best complemented by a positive analysis which takes into account the actual characteristics of the distributions, as in \cite{BeisbartBovens2008}.

Following the symmetry of the preferences, our analysis also focuses on \textit{symmetric} profiles, in which all groups use the same rule: $\phi_i=\phi$ for all $i$.
With a slight abuse of notation,  we write $\phi$ both
for a \textit{single rule} $\phi$ and for the \textit{symmetric profile} 
$(\phi,\phi,\cdots)$, which should not create confusion as long as we refer to symmetric profiles.
As for the alternatives, 
it is natural to consider that the label should not matter when the group-wide vote margin is translated into the weight allocation, given the symmetry of the preferences. 

\begin{as} \label{as:neutral}
	We assume that the rule is monotone and {neutral}, that is, $\phi$ is a non-decreasing, odd function: $\phi( \theta_i) = - \phi( - \theta_i)$.
\end{as}

Let $\pi_i(\phi;n)$
denote the expected payoff
for group $i (\leq n)$ under profile $\phi$ when the set
of groups is $\{1,\cdots,n\}$ and each group $j$'s weight is $w_j$,
the $j$th component of the sequence of weights.
The definition of $\pi_i(\phi;n)$
is the same as $\pi_i(\phi)$ in the preceding sections;
the new notation just clarifies its dependence on 
the number of groups $n$.

The main welfare criterion employed in this section is the asymptotic Pareto dominance. 

\begin{definition}
	For two symmetric profiles $\phi$ and $\psi$,
	we say that $\phi$ \textit{asymptotically Pareto dominates}
	$\psi$ if  there exists $N$ such that for all $n>N$
	and all $i=1,\cdots,n$, 
	\[
	\pi_i(\phi;n)>
	\pi_i(\psi;n).
	\]
\end{definition}

\subsection{Pareto Dominance}

The following is the main result in our asymptotic analysis.

\begin{theorem}
	Under Assumptions \ref{as:general}-\ref{as:neutral}, 
	the proportional profile asymptotically
	Pareto dominates
	all other symmetric profiles.
	In particular, it asymptotically Pareto dominates the {dominant-strategy} equilibrium of the game, {i.e.}, the symmetric winner-take-all profile.
	\label{thm:pareto}
\end{theorem}

We use the following lemma to prove Theorem \ref{thm:pareto}. The proof of Lemma \ref{lem:pareto} is relegated to the Appendix.
The proof of part (ii) uses a more general result, Lemma \ref{lem:cdm}, stated in
the next subsection, whose proof also appears in the Appendix.

\begin{lem}
	Under Assumptions \ref{as:general}-\ref{as:neutral}, the following statements hold.

	\begin{itemize}
		\item[(i)] 
		For any symmetric profile $\phi$,
		\begin{equation*}
		\begin{split}
		&\pi_i(\phi;n)\\
		&=
		2\int_0^1\theta_i \mathbb{P}\Bigg\{-w_i\phi(\theta_i)<
		\sum_{j\leq n,\,j\neq i}w_j\phi(\Theta_j)\leq
		w_i\phi(\theta_i)
		\Bigg\}dF(\theta_i).
		\end{split}
		\label{eq:pi_m}
		\end{equation*}
		
		\item[(ii)]
		For any symmetric profile $\phi$,
		as $n\to\infty$,
		\[
		\sqrt{2\pi n}
		\pi_i(\phi;n)
		\to
		2w_i
		\sqrt{
			\frac{
				\mathbb{E}[\Theta^2]
			}{
				\int_{\underline{w}}^{\bar{w}}w^2dG(w)
			}
		}
		{\rm\, Corr\,}[\Theta,\phi(\Theta)],\footnote{Since
			$\Theta$ and $\phi(\Theta)$
			are symmetrically
			distributed, the correlation
			is given by
			${\rm Corr\,} [\Theta,\phi(\Theta)]=
			\mathbb{E}[\Theta\phi(\Theta)]/
			\sqrt{\mathbb{E}[\Theta^2]\mathbb{E}[\phi(\Theta)^2]}$
			unless $\phi(\Theta)$
			is almost surely zero. If $\phi(\Theta)$ is almost surely zero,
			then the correlation
			is zero.
		}
		\]
		uniformly in $w_i\in[\underline{w},\bar{w}]$,
		where
		$\Theta$ is a random variable having the same distribution
		$F$ as $\Theta_i$.
		The limit depends on the
		profile $\phi$
		only through the factor
		${\rm Corr}[\Theta,\phi(\Theta)]$.

	\end{itemize}
	\label{lem:pareto}
\end{lem}

\begin{proof}[Proof of Theorem
	\ref{thm:pareto}]
	\hspace{1cm}\newline
	The heart of the proof is in the correlation result shown in part (ii) of Lemma \ref{lem:pareto}.
	It follows that if 
	correlation of $\phi(\Theta)$ with $\Theta$ is higher than 
	that of $\psi(\Theta)$,
	then for each group $i$, there
	exists $N_i$ such that if the number of groups ($n$) is greater than $N_i$, group $i$ ($\leq n$) will be better off under $\phi$ than $\psi$.

	Note that
	the convergence in part (ii) of Lemma \ref{lem:pareto}
	is uniform in $w_i\in[\underline{w},\bar{w}]$.
	This
	implies that the convergence 
	is uniform in $i=1,2,\cdots$.\footnote{A more detailed explanation of this step is the following.
		By Lemma \ref{lem:pareto} (i), $\sqrt{2\pi n}\pi_i(\phi;n))$ 
		asymptotically behaves as
		$2\sqrt{2\pi n}\int_0^1\theta \mathbb{P}\{-w_i\phi(\theta)<
		\sum_{j\leq n}w_j\phi(\Theta_j)\leq
		w_i\phi(\theta)
		\}dF(\theta)$, where whether 
		the sum $\sum_{j\leq n}w_j\phi(\Theta_j)$
		includes the $i$th term
		or not
		is immaterial in the limit.
		The estimate of
		$\sqrt{2\pi n}\pi_i(\phi;n)$ therefore
		has the form $f_n(w_i)$,
		where $f_n(x):=2\sqrt{2\pi n}\int_0^1\theta \mathbb{P}\{-x\phi(\theta)<
		\sum_{j\leq n}w_j\phi(\Theta_j)\leq
		x\phi(\theta)
		\}dF(\theta)$.
		Lemma \ref{lem:pareto} (ii) implies that
		$f_n(x)$ converges uniformly in $x\in[\underline{w},\bar{w}]$,
		which in turn implies that
		the convergence of $\sqrt{2\pi n}\pi_i(\phi;n)
		\approx f_n(w_i)$
		is uniform in $i=1,2,\cdots$.}
	Thus there is $N$ with the above property,
	without subscript $i$, which applies to
	all groups $i=1,2,\cdots$.
	Therefore, if
	correlation of $\phi(\Theta)$ with $\Theta$ is higher than 
	that of $\psi(\Theta)$,
	then $\phi$ asymptotically Pareto dominates $\psi$.
	
	Since the perfect correlation ${\rm Corr}[\Theta,\phi^{\rm PR}(\Theta)]=1$ is attained by the proportional rule, Theorem \ref{thm:pareto} follows. 
	%Combining with the monotonicity shown in part (iii) of Lemma \ref{lem:pareto}, we obtain the Pareto order as claimed in part (ii) of the theorem. 
\end{proof}

\begin{comment}
Theorem \ref{thm:pareto} shows the \KK{asymptotic} Pareto dominance of the symmetric proportional profile over any other symmetric profile.
Intuitively, when there are sufficiently many groups, the members' preferences are most efficiently aggregated to the social decision if the weights are allocated proportionally to the alternatives \textit{by all groups}. However, such a profile cannot be sustained as a Nash equilibrium of the game, since each group has an incentive to deviate to a dominant strategy, i.e., the winner-take-all rule. 
This typical Prisoner's Dilemma situation suggests to us that a Pareto efficient outcome is not expected to be achieved under decentralized decision making, and a coordination device is necessary in order to attain a Pareto improvement.
\end{comment}

The above results show that the winner-take-all rule is characterized by its strategic dominance, while the proportional rule is characterized by its asymptotic Pareto dominance. The following proposition provides a complete Pareto order among all the linear combinations of the two rules. 

Remember that we defined the mixed rules in Section \ref{sec:themodel} above. For $0\leq a\leq 1$, a fraction $a$ of the weight is assigned to the winner of the {group-wide vote}, while the rest, $1-a$, is distributed proportionally to each alternative:
$$
\phi^a(\theta_i)=a \phi^{\rm WTA}(\theta_i)+(1-a)\phi^{\rm PR}(\theta_i).
$$

\begin{prop}
	Under Assumptions \ref{as:general}-\ref{as:cor},	mixed profile $\phi^a$
	asymptotically Pareto
	dominates 
	mixed profile $\phi^{a^\prime}$ for any
	$0\leq a<a^\prime\leq 1$.
	In particular, the proportional profile asymptotically Pareto dominates any mixed profile $\phi^a$ for $0 < a<1$, which in turn asymptotically Pareto dominates the winner-take-all profile. In other words, all mixed profiles can be ordered by asymptotic Pareto dominance, from the proportional profile as the best, to the winner-take-all profile as the worst.	
	\label{prop:mixed}
\end{prop}

\begin{proof}
	In Appendix.
\end{proof}

The winner-take-all rule is not only asymptotically Pareto inefficient, but the worst among the symmetric mixed profiles. 
Is it worse than \textit{any} other symmetric profile? We provide an answer in Remark \ref{rem:worst} below.

\begin{rem}
	{\textit{What is the worst profile?}}
	Theorem
	\ref{thm:pareto}
	leaves the natural
	question of
	whether 
	the winner-take-all profile
	is the worst among all symmetric profiles,
	in terms of asymptotic
	Pareto dominance.
	The answer is negative.
	To see this, note first that, for the winner-take-all profile,
	the correlation
	in Lemma \ref{lem:pareto}
	is strictly positive: 
	${\rm Corr}[\Theta,\phi^{\rm WTA}(\Theta)]=
	\mathbb{E}(|\Theta|)/\sqrt{\mathbb{E}(\Theta^2)}>0$.
	On the other
	hand,
	for the symmetric profile
	$\phi^0$ in which
	the rule is defined
	by $\phi^0(\theta)=0$
	for almost all $\theta$,
	the correlation
	is obviously zero.
	This rule assigns exactly half of the weight to each alternative, regardless of the group-wide vote. 
	Thus 
	the profile $\phi^0$ is the worst among all symmetric profiles, as 
	the social decision is made by a coin toss almost surely, yielding expected payoff 0 to all groups. 
	In the rest of this section, we exclude such a trivial profile from our consideration.
%	{One can also find non-constant symmetric profiles
%	in the neighborhood of $\phi^0$
%	that are worse than $\phi^{\rm WTA}$.}
	\label{rem:worst}
\end{rem}

%Our asymptotic analysis relies on three assumptions concerning the preferences distribution: Assumption \ref{as:general} for the existence of the latent variables $(\Theta_i)_{i \in \mathbb{N}_+}$, Assumption \ref{as:pdf} for absolute continuity and full support, and Assumption \ref{as:cor} for independence and symmetry. 
%The following remark summarizes the role of them.

\subsection{Congressional District Method}
\label{subsect:cdm}

The analysis
in the preceding
subsection
suggests that 
the proportional profile is optimal
in terms of Pareto efficiency.
However,
our model
also implies that
this profile
produces an unequal distribution
of welfare;
in fact, this unequal nature
pertains to all symmetric profiles.
The Correlation Lemma \ref{lem:pareto} (ii)
shows that for these profiles,
the expected payoff for a group
is asymptotically proportional to
its weight, providing high expected payoffs to groups with a large weight.

In this subsection,
we examine whether such inequality
can be alleviated without impairing efficiency
by using an asymmetric profile, based on the Congressional District Method (CDM), currently used in Maine and Nebraska.
This profile
allocates
a fixed amount $c$
of
each group's weight
by 
the winner-take-all rule
and the rest by
the proportional rule:
\[
w_i\phi^{\rm CD}(\theta_i,w_i)=
c\phi^{\rm WTA}(\theta_i)+(w_i-c)\phi^{\rm PR}(\theta_i).
\]
We consider the \textit{congressional district profile} $\phi^{\rm CD}$ in which the rule is used by all groups.
Note that the profile is not symmetric 
in the sense that we defined at the beginning of this section.
In a mixed profile, the ratio $a$ is common to all groups. 
However, 
$\phi_i$ is not the same function of $\theta_i$ for all $i$ in $\phi^{\rm CD}$,
because the weight allocation rule depends on $w_i$, which is heterogeneous across groups. 
For example, for a large state such as California, a fixed number of weight $c$ ($=2$) implies a low ratio of $a$ ($=2/55$), while for a small state such as Wyoming, it implies a higher ratio of $a$ ($=2/3$). 
Therefore, we cannot apply Theorem \ref{thm:pareto} in order to obtain a Pareto dominance relationship. 
However, we can obtain a small-group advantage result (Theorem \ref{thm:cdm}) and a Lorenz dominance result (Theorem \ref{thm:ineq}).
%but the way $\phi^{\rm CD}$ depends on $w_i$ is the same for all groups. 
To ensure that
the profile is well-defined,
we impose 
that
the lower bound of weights $\underline{w}$ is strictly positive and
$c\in(0,\underline{w}]$.

\begin{comment}
\begin{as}
	There exists $\underline{w}>0$ such that $w_i>\underline{w}$
	for all $i$.
	\label{as:w1}
\end{as}
\end{comment}

\begin{theorem}
	Under Assumptions \ref{as:general}-\ref{as:neutral},
	let us consider the congressional district profile with parameter
	$c\leq\underline{w}$. For any symmetric profile $\phi$,
	there exists
	$w^\ast\in[\underline{w},\bar{w}]$
	with the following property:
	for any $\varepsilon>0$, there is $N$ such that for all $n>N$
	and $i=1,\cdots,n$, 
	\begin{equation*}
	\begin{split}
	&w_i<w^\ast-\varepsilon\Rightarrow
	\pi_i(\phi^{\rm CD};n)>\pi_i(\phi;n),\\
	&w_i>w^\ast+\varepsilon\Rightarrow
	\pi_i(\phi^{\rm CD};n)<\pi_i(\phi;n).
	\end{split}
	\end{equation*}
	\label{thm:cdm}
\end{theorem}

The proof of Theorem
\ref{thm:cdm}
uses the following
lemma, which shows that the correlation lemma holds for a class of profiles 
such that the weight allocation rules have the following specific form of separability.   
% in which we consider a class of "quasi-symmetric" profiles...  
Its proof and the Local Limit Theorem used in the proof are relegated
to the Appendix. 

\begin{as} \label{as:fnform}
	Let $\phi=\left(\phi_i\right)_{i=1}^\infty$ be a profile. There exist functions $h_1, h_2, h_3$ such that 
\[
w_i\phi_i(\theta_i,w_i)=h_1(w_i)h_2(\theta_i)+
h_3(w_i){\rm\,sgn\,}\theta_i, \text{ for all } i 
\]
where (i)
$h_1$ is bounded, 
(ii)
$h_2$
is an odd function such that
the support of
the distribution
of
$h_2(\Theta_i)$
contains 0, and (iii) 
$h_3$ is
continuous
but not 
constant.\footnote{Under this form, $\phi_i (\cdot, \cdot)$ is the same for all $i$ so that we can omit subscript $i$ whenever there is no confusion.}
\end{as}

It is straightforward to show that Assumption \ref{as:fnform} is satisfied for any symmetric profile %(i.e. $\phi_i=\phi$) 
as well as the congressional district profile. 
For a symmetric profile $\phi$,
let
$h_1(w_i)=w_i$,
$h_2(\theta_i)=\phi(\theta_i)-
r\,{\rm sgn\,}\theta_i$,
and $h_3(w_i)=
w_ir$ where 
$r>0$
is any positive number
in the support of the distribution of $\phi(\Theta)$.\footnote{
This is possible since $\phi(\Theta)$
is symmetrically distributed,
and since we exclude the trivial case
in which $\phi(\Theta)=0$ almost surely.} 
For the congressional
district profile $\phi^{\rm CD}$,
let
$h_1(w_i)=w_i-c$,
$h_2(\theta_i)=\theta_i-\,{\rm sgn\,}\theta_i$,
and
$h_3(w_i)=
w_i$.

\begin{lem}
	Under Assumptions \ref{as:general}-\ref{as:neutral},
	let $\phi$ be a profile which satisfies Assumption \ref{as:fnform}.
Then,
	as $n\to\infty$,
	\begin{equation*}
	\sqrt{2\pi n}
	\pi_i(\phi;n)
	\to
	\frac{2w_i\mathbb{E}[\Theta\phi(\Theta,w_i)]}{
		\sqrt{\int_{\underline{w}}^{\bar{w}}w^2\mathbb{E}[\phi(\Theta,w)^2]dG(w)}},
%	\footnote{
%		This formula
%		implicitly
%		excludes the
%		case where
%		the denominator
%		is zero.
%		Note that
%		in any such case,
%		$\phi(\theta_i,w_i)=0$
%		for almost every $\theta_i\in[-1,1]$
%		and $w_i\in[\underline{w},\bar{w}]$.
%		This means that
%		the social decision
%		is almost always
%		made by a coin toss,
%		and the expected payoff is $0$,
%		implying that the limit
%		is trivially zero. 
%	%\label{ftn:zeroas}
%}
	\end{equation*}
	uniformly in $w_i\in[\underline{w},\bar{w}]$,
	where $\Theta$ is a random variable having the same distribution
	$F$ as $\Theta_i$.
	\label{lem:cdm}
\end{lem}

\begin{proof}[Proof of Theorem
	\ref{thm:cdm}]
	By Lemma \ref{lem:cdm},
	the expected payoff
	for group $i$ under a symmetric profile $\phi$ tends
	to a linear function of 
	$w_i$. Let $A^\phi$ be the coefficient:
	\begin{equation}
	\begin{split}
	\lim_{n\rightarrow \infty }\sqrt{2\pi n}
	\pi_i(\phi;n)
	&=\frac{2w_i\mathbb{E}[\Theta \phi(\Theta)]}{
		\sqrt{\mathbb{E}[\phi(\Theta)^2]
			\int_{\underline{w}}^{\bar{w}}
			w^2dG(w)
		}
	}\\
	&=:A^{\phi} w_i.
	\end{split}
	\label{eq:aw_phi}
	\end{equation}
	
	For the congressional district profile, remember the definition: 
	\begin{eqnarray*}
		w_j\phi^{\rm CD}\left(\theta_j,w_j\right) &=&c\phi ^{\rm WTA}\left( \theta_j
		\right) +\left( w_j-c\right) \phi ^{\rm PR}\left( \theta_j\right) 
		\notag \\
		&=&c\ \text{sgn}\left( \theta_j\right) +\left( w_j-c\right) \theta_j.
		%		\label{eq:wi_phiCD}
	\end{eqnarray*}%
	We claim that the limit function is affine in $w_i$:%
	\begin{equation}
	\lim_{n\rightarrow \infty }\sqrt{2\pi n}
	\pi_i(\phi^{\rm CD};n)=Bw_i+C.
	\label{eq:bw+c}
	\end{equation}%
	To see that, let us apply Lemma \ref{lem:cdm} again:
	\begin{eqnarray*}
		\lim_{n\rightarrow \infty }\sqrt{2\pi n}
		\pi_i(\phi^{\rm CD};n) 
		&=&2\cdot\frac{w_i\mathbb{E}\left[ \Theta\phi^{\rm CD}\left(\Theta,w_i\right) \right] }{\sqrt{\int_{\underline{w}}^{\bar{w}}w^2 \mathbb{E}\left[
				\phi^{\rm CD}\left(\Theta,w\right) ^{2}\right] dG(w)}} \\
		&=&2\cdot\frac{c\mathbb{E}\left[ \left\vert \Theta \right\vert \right] +\left(
			w_{i}-c\right) \mathbb{E}\left[ \Theta ^{2}\right] }{\sqrt{\int_{\underline{w}}^{\bar{w}}w^2 \mathbb{E}\left[
				\phi^{\rm CD}\left(\Theta,w\right) ^{2}\right] dG(w) }}.
	\end{eqnarray*}%
	Since $\left\vert \theta \right\vert \geq \theta ^{2}$ with a strict
	inequality for $0<\left\vert \theta \right\vert <1,$ 
	the full support condition for $\Theta$
	implies $\mathbb{E}\left[
	\left\vert \Theta \right\vert \right] >\mathbb{E}\left[ \Theta ^{2}\right]$,
	and thus the intercept $C$ is positive. The coefficient of $w_i$
	is: 
	\begin{equation*}
	B=\frac{2\mathbb{E}\left[ \Theta ^{2}\right] }{\sqrt{\int_{\underline{w}}^{\bar{w}}w^2 \mathbb{E}\left[
			\phi^{\rm CD}\left(\Theta,w\right) ^{2}\right] dG(w) }}.
	\end{equation*}
	If $A^{\phi} <B$, combined with $C>0$, the right-hand side of (\ref{eq:bw+c}) is above that of (\ref{eq:aw_phi}). Then, set $w^*=\bar{w}$.
	If $A^{\phi} >B$, again combined with $C>0$, the two limit functions (\ref{eq:aw_phi}) and (\ref{eq:bw+c}) intersect only once at a
	positive value $\hat{w}$. Let $w^*=\max\left\{ \underline{w} ,\min \{\hat{w},\bar{w}\}\right\}$.
	
	%	Now, remember that $\left\vert \phi ^{WTA}\left( \theta \right)
	%	\right\vert =\left\vert \text{sgn}\left( \theta\right) \right\vert \geq
	%	\left\vert \theta \right\vert =\left\vert \phi ^{PR}\left( \theta
	%	\right) \right\vert $ with a strict inequality for $0<\left\vert \theta
	%	\right\vert <1$. Since $\phi^{CD}$ is a weighted average of $\phi
	%	^{WTA}$ and $\phi ^{PR}$, we have 
	%	$\left\vert \phi^{PR}\right\vert \leq
	%	\left\vert \phi^{CD}\right\vert \leq
	%	\left\vert \phi ^{WTA}\right\vert $ 
	%	with a strict inequality for $0<\left\vert
	%	\theta\right\vert <1$. Therefore, 
	%	$
	%	\mathbb{E}[\phi^{PR}(\Theta)^2]
	%	\int_0^{\bar{w}}
	%	x^2dG(x)
	%	<
	%	\int_0^{\bar{w}}x^2 \mathbb{E}\left[
	%	\phi^{CD}\left(\Theta,x\right) ^{2}\right] dG(x) 
	%	<
	%	\mathbb{E}[\phi^{WTA}(\Theta)^2]
	%	\int_0^{\bar{w}}
	%	x^2dG(x),$ 
	%	which implies $A^{\phi^{WTA}}<B<A^{\phi^{PR}}$.
	%	Combined with $C>0$, the right-hand side of (\ref{eq:bw+c}) is above that of (\ref{eq:aw_phi}) for $\phi^{WTA}$, which proves part (i). Between $\phi^{PR}$ and $\phi^{CD}$, two limit functions intersect only once at a
	%	positive value $w^{\ast }$. 
	Since the convergences
	(\ref{eq:aw_phi})
	and (\ref{eq:bw+c})
	are uniform in $w_i$,
	for any $\varepsilon>0$
	there is $N$ with the property
	stated in Theorem \ref{thm:cdm}.
\end{proof}

Theorem \ref{thm:cdm}
implies that
the congressional district profile
makes the members
of groups with small weights
better off, compared with \textit{any} symmetric profile.
If
the weight
is an
increasing function
of the group size, it means that the congressional district profile is favorable for the members of small groups.\footnote{
As a special case, we cannot rule out the possibility where $w^*$ is equal to $\underline{w}$ so that $\phi^{\rm CD}$ is Pareto dominated by $\phi$.
However, this can only happen when $A^\phi$ is greater than $B$, which implies that the ratio $\pi_i(\phi^{\rm CD})/\pi_i(\phi)$ is decreasing with respect to $w_i$ (see (\ref{eq:aw_phi}) and (\ref{eq:bw+c})).
Thus, even in such a case, CDM favors groups with small weights in terms of relative comparison of payoffs.}

The
intuitive
reason
why the congressional
district profile
is advantageous
for small groups
is as follows.
Under this profile,
the ratio
of weights
cast
by the winner-take-all
rule ({i.e.}, $c/w_i$)
is higher
for small
groups
than for large
groups.
Therefore,
the rules used by the smaller groups 
are relatively close to the dominant strategy, inducing 
a relative advantage
for the
small groups.
We provide a numerical result in the following subsection using an example of the US Electoral College. 

\bigskip

%\begin{rem}
In addition
to Theorem
\ref{thm:cdm},
we can also show that the
congressional district profile
{distributes payoffs
more equally}
than any symmetric profile does,
in the sense of
Lorenz dominance.
{A profile
of group
payoffs,
$\pi=(\pi_1,\cdots,\pi_n)$,
is said to
\textit{Lorenz dominate}
another profile, $\pi^\prime=(\pi_1^\prime,\cdots,
\pi_n^\prime)$,
if the share of payoffs
acquired by any bottom fraction
of groups is larger in $\pi$ than in $\pi^\prime$.}\footnote{Here Lorenz dominance is defined for distributions of payoffs among \textit{groups}.
If a micro foundation can be provided by taking the group payoff as the average payoff of the members who are ex ante symmetric (as in Remark \ref{rem:latent}), then the definition is equivalent to Lorenz dominance for distributions of \textit{individual} payoffs.}
\begin{comment}\footnote{{Formally,
	for each $x\in[-1,1]$,
	let
	$H_\pi(x)$
	be the total population share
	of those groups
	whose per capita welfare
	is not greater than $x$
	under the payoff profile $\pi$.
	Then $H_\pi$ is a distribution
	function.
	The \textit{Lorenz curve} of $H_\pi$
	is the graph of the function
	$\int_0^{H_\pi^{-1}(p)}x dH_\pi(x)/
	\int_0^{1}x dH_\pi(x)$,
	$0\leq p\leq1$,
	where we define
	$H_\pi^{-1}(p)=\sup\{x | H_\pi(x)\leq p\}$.
	A payoff profile $\pi$
	\textit{Lorenz dominates}
	another profile $\pi^\prime$ if
	the Lorenz curve
	of $H_\pi$ lies above
	that of $H_{\pi^\prime}$.}}
	\end{comment}
Lorenz dominance,
whenever it occurs,
agrees
with equality
comparisons
by various 
inequality indices
including
the coefficient of variation,
the Gini coefficient,
the Atkinson index,
and the Theil index
(see \citet{FieldsFei1978} 
and \citet{Atkinson1970}).
To see why
the congressional district profile
is more equal than
any symmetric profile, recall
equations
(\ref{eq:aw_phi})
and (\ref{eq:bw+c})
in the proof
of Theorem
\ref{thm:cdm},
which assert
that when the number
of groups is large,
the payoff
for group $i$ is approximately
$A^{\phi} w_i$
for the symmetric
profile,
and it is approximately
$Bw_i+C$
for the congressional
district profile.
The constant
term $C>0$
for the congressional district
profile 
assures equal additions
to all groups' payoffs,
which results in
a more equal distribution
than when there is no such term.
More precisely,
we can prove the
following statement.
The proof is relegated
to the Appendix.
\label{rem:cdm}
%\end{rem}

\begin{theorem}
	%Fix $n$ as the number of groups.
	Under Assumptions \ref{as:general}-\ref{as:neutral},	
	let us consider the
	payoff profile
	under the congressional district
	profile:
	$\pi\left(\phi^{\rm CD};n\right)=\left(\pi_i \left(\phi^{\rm CD};n\right) \right)_{i=1}^n$.
	Let $\phi$ be 
	any symmetric profile
	and
	$\pi\left(\phi;n\right)=\left(\pi_i \left(\phi;n\right)\right)_{i=1}^n$
	the payoff profile under $\phi$.
	For sufficiently
	large $n$,
	$\pi\left(\phi^{\rm CD};n\right)$
	Lorenz dominates
	$\pi\left(\phi;n\right)$.
	\label{thm:ineq}
\end{theorem}

\subsection{Computational Results}
\label{sec:comp}

The 
results in the previous subsection concern cases with a large number
of groups.
The question remains
as to whether
the conclusions 
obtained there
are also valid
for a finite number
of groups.
In this section,
we 
provide a numerical computation result using an example
of the US
presidential election.

There
are 50
states
and one federal district.
The 
weights
$(w_i)_{i=1}^{51}$
are the 
numbers
of electoral votes
assigned in the 2020 election.
%This number
%equals the
%state's total number
%of seats in the Senate
%and House of Representatives.
%Thus, $w_i$
%is two plus
%a number
%that is roughly
%proportional to 
%the state's
%population.
The first and second
columns 
of
Table
\ref{t:pi}
show
the distribution
of weights
among the states.

We assume IAC* (Impartial Anonymous Culture*): 
the statewide popular vote margins $\Theta_i$ are independent and uniformly distributed on $[-1,1]$,
first introduced by
\cite{May1948}
and studied thoroughly by, for example, 
\cite{DeMouzon2019}. For any profile $\phi$, 
we can compute the
per capita payoff
for state $i$
via the formula:
\begin{equation}
\pi_i(\phi)
=2 \left( 0.5^{51} \right)
\int_{-1}^1\cdots\int_{-1}^1
\theta_i1_{A}(\theta_1,\cdots,\theta_{51})
d\theta_1\cdots d\theta_{51}
\label{eq:pi_u}
\end{equation}
where
$A=\left\{ (\theta_1,\cdots,\theta_{51}) \left| \sum_{j=1}^{51}w_j\phi_j(\theta_j)>0 \right. \right\}$.\footnote{It is easy to check that
	under the uniform distribution assumption,
	(\ref{eq:pi_u}) is equivalent to the expression
	in Lemma \ref{lem:pareto} (i).}

We consider 
four distinct profiles:
$\phi^{\rm WTA}$,
$\phi^{\rm PR}$,
$\phi^a$
with $a=102/538$,
and $\phi^{\rm CD}$
with coefficient $c=2$,
which are the winner-take-all profile,
the proportional profile,
a mixed profile,
and 
a congressional district
profile, respectively.
The parameter
$c=2$
of
the congressional
district
profile
is
the number
used
in Maine
and Nebraska,
corresponding to the two seats
assigned to each state
in the Senate.
The parameter
$a=102/538$
of the
mixed profile
is chosen
so that
the
proportion
of electoral votes
allocated on the winner-take-all basis
is the same
for all states,
and
the total number
of electoral
votes allocated in this way
is the same
as in the congressional
district profile.

We compute (\ref{eq:pi_u})
under these four profiles
by a Monte Carlo simulation with $10^{10}$ iterations. 
The results
are
summarized
in Tables
\ref{t:pi} and \ref{t:pi_ratio}.
Table \ref{t:pi}
shows the
per capita payoff
($\pi_i(\phi)$)
under the respective profiles.
%The estimated error 
%in computing
%each of these probabilities
%is within
%one percent of the computed value.
Table
\ref{t:pi_ratio}
shows
the ratios
of per capita payoff
between different
profiles
($\pi_i(\phi)/\pi_i(\psi)$).
If the ratio is below $1$,
state $i$ prefers $\psi$ to $\phi$.

It follows from
Lemma \ref{lem:pareto} (ii) that
as the number $n$
of states increases,
the ratios $\pi_i \left( \phi^{\rm WTA}\right) / \pi_i \left( \phi^{\rm PR} \right)$
and $\pi_i \left(\phi^{a}\right) / \pi_i\left(\phi^{\rm PR}\right)$
converge to the respective
correlations
$\mathrm{Corr}[\Theta,\phi^{\rm WTA}(\Theta)]\approx0.866$
and $\mathrm{Corr}[\Theta,\phi^{a}(\Theta)]\approx
0.989$, where
the values are computed for $\Theta$
uniformly
distributed
on $[-1,1]$.
Table
\ref{t:pi_ratio}
indicates that
for the present
example with 50
states plus DC,
these ratios
are indeed close
to the respective
correlations,
which suggests
that convergence of the
$\pi$-ratios
is fairly quick.
In particular,
as 
expected by Theorem
\ref{thm:pareto},
the proportional profile
Pareto dominates
the winner-take-all profile
in the present case.
As suggested by Proposition \ref{prop:mixed}, 
all states
prefer the mixed profile $\phi^a$ to the winner-take-all profile,
and the proportional profile to $\phi^a$.

The ratios
$\pi_i(\phi^{\rm CD})/
\pi_i(\phi^{\rm PR})$
in Table \ref{t:pi_ratio} are consistent with the result in Theorem \ref{thm:cdm}.
Small 
states
prefer
the congressional district
profile to the proportional one.

In addition, 
the values
of $\pi_i(
\phi^{\rm CD})/
\pi_i(\phi^{\rm WTA})$
in the table
show that
the winner-take-all profile
is Pareto dominated by the congressional district profile,
and the
welfare
improvement
by switching
to the congressional district profile
is greater for small states
than for large states in terms of the ratio.

All of our numerical observations suggest the sensibility of the asymptotic results obtained in Section \ref{sec:asym} in the example of the US Electoral College.

\newpage

\begin{table}[h]
	\centering
	\caption{Estimated payoffs in the US presidential election, based on the apportionment in 2016, via Monte Carlo simulation with $10^{10}$ iterations. The estimated standard errors are in the range between $3.9$ and $4.1 \times 10^{-6}$.}
	\begin{tabular}{llllll} \hline\hline
		electoral & number & $\pi(\phi^{\rm WTA})$ & $\pi(\phi^{\rm PR})$ & $\pi(\phi^{a})$ & $\pi(\phi^{\rm CD})$ \\ 
		votes & of states & & & &  \\ \hline
3  & 8 & 0.0113 & 0.0133 & 0.0130 & 0.0167 \\
4  & 5 & 0.0151 & 0.0177 & 0.0174 & 0.0209 \\
5  & 3 & 0.0189 & 0.0221 & 0.0217 & 0.0251 \\
6  & 6 & 0.0226 & 0.0266 & 0.0261 & 0.0293 \\
7  & 3 & 0.0264 & 0.0310 & 0.0305 & 0.0335 \\
8  & 2 & 0.0302 & 0.0354 & 0.0348 & 0.0377 \\
9  & 3 & 0.0340 & 0.0399 & 0.0392 & 0.0419 \\
10 & 4 & 0.0378 & 0.0443 & 0.0436 & 0.0461 \\
11 & 4 & 0.0416 & 0.0488 & 0.0479 & 0.0503 \\
12 & 1 & 0.0454 & 0.0532 & 0.0523 & 0.0545 \\
13 & 1 & 0.0492 & 0.0577 & 0.0567 & 0.0587 \\
14 & 1 & 0.0531 & 0.0622 & 0.0611 & 0.0630 \\
15 & 1 & 0.0569 & 0.0666 & 0.0655 & 0.0672 \\
16 & 2 & 0.0607 & 0.0711 & 0.0699 & 0.0715 \\
18 & 1 & 0.0684 & 0.0801 & 0.0788 & 0.0800 \\
20 & 2 & 0.0762 & 0.0891 & 0.0877 & 0.0885 \\
29 & 2 & 0.1120 & 0.1303 & 0.1284 & 0.1275 \\
38 & 1 & 0.1494 & 0.1729 & 0.1706 & 0.1677 \\
  55 & 1 & 0.2356 & 0.2614 & 0.2615 & 0.2507
		\\ \hline
	\end{tabular}
	\label{t:pi}
\end{table}

\newpage

\begin{table}[h]
	\caption{Ratios between payoffs.}
	\centering
	\begin{tabular}{llllll} \hline\hline
		electoral & number & $\frac{{\pi}(\phi^{\rm WTA})}
		{{\pi}(\phi^{\rm PR})}$ & 
		$\frac{{\pi}(\phi^{a})}{{\pi}(\phi^{\rm PR})}$ & 
		$\frac{{\pi}(\phi^{\rm CD})}{{\pi}(\phi^{\rm PR})}$ &
		$\frac{{\pi}(\phi^{\rm CD})}{{\pi}(\phi^{\rm WTA})}$\\ 
		votes & of states & & & &\\ \hline
		3  & 8 & 0.852 & 0.982 & 1.260 & 1.479 \\
		4  & 5 & 0.852 & 0.982 & 1.182 & 1.387 \\
		5  & 3 & 0.852 & 0.982 & 1.134 & 1.331 \\
		6  & 6 & 0.852 & 0.982 & 1.103 & 1.294 \\
		7  & 3 & 0.852 & 0.982 & 1.080 & 1.268 \\
		8  & 2 & 0.852 & 0.982 & 1.064 & 1.248 \\
		9  & 3 & 0.852 & 0.982 & 1.050 & 1.232 \\
		10 & 4 & 0.853 & 0.983 & 1.040 & 1.220 \\
		11 & 4 & 0.853 & 0.983 & 1.031 & 1.210 \\
		12 & 1 & 0.853 & 0.983 & 1.024 & 1.201 \\
		13 & 1 & 0.853 & 0.983 & 1.018 & 1.194 \\
		14 & 1 & 0.853 & 0.983 & 1.013 & 1.187 \\
		15 & 1 & 0.854 & 0.983 & 1.009 & 1.181 \\
		16 & 2 & 0.854 & 0.983 & 1.005 & 1.177 \\
		18 & 1 & 0.854 & 0.983 & 0.998 & 1.168 \\
		20 & 2 & 0.855 & 0.983 & 0.993 & 1.161 \\
		29 & 2 & 0.859 & 0.985 & 0.978 & 1.138 \\
		38 & 1 & 0.864 & 0.987 & 0.970 & 1.122 \\
		55 & 1 & 0.901 & 1.000 & 0.959 & 1.064
		\\ \hline
	\end{tabular}
	\label{t:pi_ratio}
\end{table}

\section{Concluding Remarks}

This paper shows that the decentralized choice of the weight allocation rule in representative democracy constitutes a Prisoner's Dilemma: the winner-take-all rule is a dominant strategy for each group, whereas the Nash equilibrium is Pareto dominated. 
Each group has an incentive to put its entire weight on the alternative supported by the majority of its members in order to reflect their preferences in the social decision, although 
such a distortion by each group prevents
efficient aggregation of the preferences of the society as a whole.
%it fails to efficiently aggregate the preferences of all members in the society.
%{The society which consists of distinct groups thus faces a dilemma between group-level aggregation and mutual optimality for the whole society.}

We also develop an asymptotic technique and show that the proportional rule Pareto dominates every other symmetric profile when the number of the groups is sufficiently large.

Our model may provide explanations for the phenomena that we observe in existing institutions of collective decision making. In the United States Electoral College, the rule used by the states varied in early elections until it converged by 1832 to the winner-take-all rule, which has remained dominantly employed by nearly all states since then. In many parliamentary voting situations, we often observe parties and/or factions forcing their members to align their votes in order to maximally reflect their preferences in the legislative decision, although some members may disagree with the party's alignment. The voting outcome obtained by the winner-take-all rule may fail to efficiently aggregate preferences, as observed in the discrepancy between the electoral result and the national popular vote winner in the US presidential elections in 2000 and 2016. 
Party discipline or factional voting may also cause welfare loss when each group pushes their votes maximally toward their ideological goals, failing to reflect all of their members' preferences in the legislative decision.

The Winner-Take-All Dilemma tells us that the society should call for some device other than each group's unilateral effort, in order to obtain a socially preferable outcome. As we see in the failure of various attempts to modify or abolish the winner-take-all rule, such as the ballot initiative for an amendment to the State Constitution in Colorado in 2004, each state has no incentive to unilaterally deviate from the equilibrium. 
The National Popular Vote Interstate Compact is a well-suited example of a coordination device (\cite{Koza2013}). As it comes into effect only when the number of electoral votes attains the majority, each state does not suffer from the payoff loss by a unilateral (or coalitional) deviation until a sufficient level of coordination is attained. The emergence of such an attempt is coherent with the insights obtained in this paper that the game is a Prisoner's Dilemma, and a coordination device is necessary for a Pareto improvement.

%Our main result (Theorem \ref{thm:dilemma}) is abstract in that we do not impose specific assumptions on the preferences distribution based on the observed characteristics in the real representative voting problems. 
%Additionally, we impose an impartiality assumption in our asymptotic analysis. 
%Obviously, our normative analysis would be best complemented by a positive analysis, which we leave for future research.

We have assumed that social decisions are binary. 
When the model is applied to a yes-or-no collective decision, as in the case of legislative voting, the assumption of a binary decision can be justified on the grounds that choices are made between the status quo and a proposal. 
However, such an argument abstracts away the political process that gives rise to the particular choice of the proposal or the endogenous party formation. 
The analysis would be inherently different in models with three or more alternatives, primarily because group-level incentives for aggregation rules depend heavily on individual-level incentives for strategic voting. 
Therefore, in order to provide further analysis, additional assumptions are needed about factors such as individual voting behavior or the alternative space.
A successful analysis is provided by \cite{Kurz2017}, for example, in which a one-dimensional structure is introduced in the set of alternatives. 
But a thorough analysis with three or more alternatives is beyond the scope of this paper and we leave it for future research.

\begin{comment}
We have assumed
that social decisions are binary.
There are situations where
this assumption may not fit.
In the US presidential elections,
third-party or
independent candidates
can, and do, have a non-negligible
impact on the election outcome.
It is not clear
how the presence of such candidates
alters the comparison of 
rules to allocate electoral votes.
When the model
is applied to legislative voting,
the assumption of a binary decision might be justified
on the grounds that
choices are made between
the status quo and a proposal.
However, such an argument
abstracts
away the political process that gives rise to
the particular choice of the proposal.
Further analysis is necessary for
the cases with more than two alternatives
and is beyond the scope of this paper.
\end{comment}

\phantomsection
\addcontentsline{toc}{section}{Appendix}

\section*{Appendix}

\renewcommand{\thesubsection}{A\arabic{subsection}}
\setcounter{subsection}{0}

%(For Online Publication)

\subsection{Proof
	of Lemma \ref{lem:char}}

	(i)	
	Let $D$
	be the set
	of all SCFs,
	and $\pi(D)=\{
	(\pi_i(d))_{i=1}^n|
	d\in D
	\}$
	the set
	of (ex ante)
	payoff vectors
	generated by
	SCFs.
	Then
	$\pi(D)$
	is convex.\footnote{
		This is because
		for any two
		SCFs $d$ and $d^\prime$,
		any convex
		combination
		of the
		payoff vectors
		corresponding to
		$d$ and $d^\prime$
		can be realized
		as a compound
		SCF that randomizes
		between
		$d$ and $d^\prime$.
	}
	Let
	$\mathrm{Pa}\,
	(\pi(D))$
	be the Pareto
	frontier of $\pi(D)$,
	{i.e.}, the set
	of payoff profiles
	$u\in\pi(D)$
	for which
	there is no
	$d^\prime\in\pi(D)$
	such that
	$u_i^\prime\geq
	u_i$
	for all $i$
	and $u_i^\prime>
	u_i$ for some
	$i$.

	We divide the proof of (i)
	into two steps.

	\bigskip
	
	\noindent
	\textbf{Step 1.}
	\textit{Let
		$\lambda\in
		\mathbb{R}_+^n\setminus
		\{\textbf{0}\}$.
		Then
		the unique solution
		to
		the following maximization
		problem (\ref{eq:max})
		is the payoff vector
		$u^\lambda:=
		\left(\pi_i(d^\lambda)\right)_{i=1}^n$
		under a cardinal $\lambda$-weighted
		majority rule $d^\lambda$.}\footnote{Recall that
	a weighted majority
	rule with a given weight vector
	is unique only up to differences
	on a set of measure zero, inducing
	the same payoffs.}
	\begin{equation}
	\max_{u\in \pi(D)}
	\sum_{i=1}^n\lambda_i
	u_i.
	\label{eq:max}
	\end{equation}
	\textit{Moreover,
		an SCF $d$
		satisfies
		$\left(\pi_i(d)\right)_{i=1}^n=
		u^\lambda$
		if and only
		if $d$ is a $\lambda$-weighted majority
		rule.}
	
	Let $d\in D$ be any SCF.
	Then 
	\begin{equation}
	\sum_{i=1}^n
	\lambda_i\pi_i(d)
	=\sum_{i=1}^n
	\lambda_i
	\mathbb{E}\left[
	\Theta_id(
	\Theta)
	\right]
	=
	\mathbb{E}\left[
	d(\Theta)\sum_{i=1}^n\lambda_i
	\Theta_i
	\right].
	\label{eq:weighted_sum}
	\end{equation}
	Since 
	$\Theta$
	is absolutely
	continuous,
	and so
	$\sum_{i=1}^n
	\lambda_i\Theta_i
	\neq0$
	almost surely,
	$d$ maximizes
	(\ref{eq:weighted_sum})
	if and only
	if
	$d(\Theta)
	=\mathrm{sgn\,}
	\sum_{i=1}^n\lambda_i\Theta_i$
	almost surely.
	That is, 
	\begin{equation}
	\text{$d$
		maximizes (\ref{eq:weighted_sum}) 
		$\Leftrightarrow$
		$d$ is a $\lambda$-weighted
		majority rule.}
	\label{eq:max_equiv}
	\end{equation}
	This implies the
	first sentence
	of Step 1.
	Result
	(\ref{eq:max_equiv})
	also implies that 
	if $d$ is not
	a $\lambda$-weighted
	majority rule,
	then
	$\pi_i(d)\neq
	\pi_i(d^\lambda)$
	for some $i$,
	which proves
	the ``only if''
	part of 
	the second sentence
	of Step 1.
	The ``if'' part
	is trivial.

	\bigskip
	\noindent
	\textbf{Step 2.}
	\textit{A payoff vector $u\in \pi(D)$ is in the Pareto 
		frontier  $\mathrm{Pa}\,
		(\pi(D))$
		if and only
		if there exists
		$\lambda\in
		\mathbb{R}_+^n
		\setminus\{\textbf{0}\}$
		such that
		$u=(\pi_i(d^\lambda))_{i=1}^n
		=:u^\lambda$,
		where $d^\lambda$
		is a $\lambda$-weighted majority
		rule.}
	
	Since $\pi(D)$
	is convex, we can apply
	\citet[
	Proposition 16.E.2]{MWG1995}
	to show the
	``only if''
	part of
	Step 2.

	To show the
	``if'' part,
	suppose on the contrary
	that
	$u^\lambda\notin
	\mathrm{Pa}\,
	(\pi(D))$
	for some  $\lambda\in\mathbb{R}_+^n
	\setminus\{\textbf{0}\}$.
	Then there
	exists
	$u\in \pi(D)$
	such that 
	$u\neq u^\lambda$
	and
	$u_i\geq u^\lambda_i$
	for all $i$.
	Then
	$\sum_{i=1}^n\lambda_i
	u_i\geq
	\sum_{i=1}^n\lambda_i
	u_i^\lambda$.
	This
	contradicts
	the fact that $u^\lambda$
	is the unique
	solution to problem
	(\ref{eq:max}).

	\bigskip
	(ii)
	This follows from
	the trivial fact that the set of SCFs $d_\phi$
	induced by 
	profiles $\phi$ that are equivalent to
	a generalized proportional profile
	coincides with the
	set of all weighted majority rules.\footnote{
	Indeed, if $\phi$ is equivalent
	to a generalized proportional profile
	with the vector of coefficients
	$\lambda\in[0,1]^n\setminus\{\textbf{0}\}$,
	the induced SCF $d_\phi$ is a  $\mu$-weighted majority
	rule, where the weights are defined by
	$\mu_i:=w_i\lambda_i$;
	conversely, if $d$
	is a 
	$\mu$-weighted majority rule,
	then $d=d_\phi$ for some profile $\phi$
	that is
	equivalent to the generalized proportional
	profile with coefficients
	$\lambda_i:=\frac{\mu_i}{w_i}$.}\qed

\subsection{Proof
	of Part (i) of Lemma \ref{lem:pareto}}

We prove the statement for group 1.
Let 
$\pi_1(\phi;n|\theta_1)$
be the conditional 
expected payoff for group 1
given that the group-wide margin is $\Theta_1=\theta_1$,
which by (\ref{eq:conditional_welfare})
is:
\begin{equation*}
\pi_1(\phi;n|\theta_1)=
\theta_1
(\mathbb{P}\{w_1\phi(\theta_1)+S_{\phi_{-1}}>0\}
-\mathbb{P}
\{w_1\phi(\theta_1)+S_{\phi_{-1}}<0\}).
\end{equation*}
Since $S_{\phi_{-1}}$
is symmetrically distributed,
the second
probability
can be written as $\mathbb{P}\{-w_1\phi(\theta_1)+S_{\phi_{-1}}>0\}$.
Thus,
for $\theta_1\in[0,1]$,
the above expression
equals
\[
\pi_1(\phi;n|\theta_1) =
\theta_1
\mathbb{P}\{-w_1\phi(\theta_1)<S_{\phi_{-1}}\leq w_1\phi(\theta_1)\}.
\]
By symmetry,
twice
the integral
of this
expression
over $\theta_1\in[0,1]$
(instead of $[-1,1]$)
equals the unconditional
expected payoff $\pi_1(\phi;n)$,
which proves part (i) of Lemma \ref{lem:pareto}.\qed

\subsection{Local Limit Theorem}

We quote a version of the Local Limit Theorem
shown in
\citet{Mineka1970}.
We will use it in the proof of
part (ii) of Lemma \ref{lem:pareto}.

\begin{llt}
	(\citet[Theorem 1]{Mineka1970})
	Let $(X_i)$ be a sequence of independent random variables
	with mean 0 and variances $0<\sigma_i^2<\infty$.
	Write $F_i$ for the distribution of $X_i$.
	Write also $S_n=\sum_{i=1}^nX_i$ and $s_n^2=\sum_{i=1}^n\sigma_i^2$.
	Suppose the sequence $(X_i)$ satisfies the following conditions:
	
	\begin{itemize}
		\item[($\alpha$)] There exists $\bar{x}>0$ and $c>0$ such that
		for all $i$,
		\[
		\frac{1}{\sigma_i^{2}}\int_{|x|<\bar{x}}x^2dF_i(x)>c.
		\]
		
		\item[($\beta$)] 
		Define the set
		\[
		A(t,\varepsilon)=
		\{
		x| \ |x|<\bar{x}\text{ and } |xt-\pi m|>\varepsilon
		\text{ for all integers $m$ with $|m|<\bar{x}$}
		\}.
		\]
		Then,
		for some bounded
		sequence $(a_i)$ such that
		$\inf_i\mathbb{P}\{|X_i-a_i|<\delta\}>0$ for all $\delta>0$,
		and for any $t\neq0$, there exists $\varepsilon>0$ such that
		\[
		\frac{1}{\log s_n}
		\sum_{i=1}^n
		\mathbb{P}\{X_i-a_i\in A(t,\varepsilon)\}\to \infty.
		\]
		
		\item[($\gamma$)]
		(Lindeberg's condition.)
		For any $\varepsilon>0$,
		\[
		\frac{1}{s_n^2}\sum_{i=1}^n\int_{|x|/s_n>\varepsilon}
		x^2dF_i(x)\to0.
		\]
	\end{itemize}
	\noindent
	Under conditions ($\alpha$)-($\gamma$), if $s_n^2\to\infty$, we have
	\begin{equation}
	\sqrt{2\pi s_n^2}
	\mathbb{P}\{S_n\in(a,b]\}\to
	b-a.\footnote{The original conclusion 
		of Theorem 1 in \citet{Mineka1970} is stated in terms of
		the open interval $(a,b)$.
		Applying the theorem
		to $(a,b+c)$ and $(b,b+c)$
		and then taking the difference gives the result for
		$(a,b]$. In addition,
		the original statement
		allows for cases where $s_n^2$
		does not go to infinity,
		and also mentions uniform convergence.
		These considerations are not necessary
		for our purpose, so we omit them.}
	\label{eq:llt_mineka}
	\end{equation}
	\label{llt:mineka}
\end{llt}

\subsection{Proof
	of Lemma \ref{lem:cdm}}

\noindent
\textbf{Preliminaries.}
We prove the lemma
for group 1.
In the proof, we use the notation of
LLT.
Let 
\[
X_i:=w_i\phi(\Theta_i,w_i),\,i=1,2,\cdots,
\]
and $S_n:=\sum_{i=1}^n X_i$.
Then
$X_i$ has mean 0 and variance $\sigma_i^2:=
w_i^2\mathbb{E}[\phi(\Theta,w_i)^2]$,
and so the partial sum of variances is  $s_n^2:=\sum_{i=1}^nw_i^2\mathbb{E}[\phi(\Theta,
w_i)^2]$,
where
$\Theta$ represents a random variable
that has the same distribution $F$
as $\Theta_i$.

Define the event
\[
\Omega_n(\theta_1,w_1)=\left\{-w_1\phi(\theta_1,w_1)<
\sum_{i=2}^n
X_i
\leq w_1\phi(\theta_1,w_1)\right\}.
\]

We divide the proof
into several claims.
%Claim \ref{claim:profiles}
%shows that all symmetric profiles
%and the congressional district profile
%are special cases of
%the profile $\phi$
%in Lemma \ref{lem:cdm}.
Claims \ref{claim:s_n}-\ref{claim:beta}
show that
the sequence
$(X_i)$
defined above satisfies
the conditions
of the Local Limit Theorem (LLT)
in Section A4.
Claim \ref{claim:lem_cdm}
applies the LLT
to complete
the proof of Lemma \ref{lem:cdm}.

\renewcommand{\theclaima}{5.\arabic{claima}}

%\begin{claima}
%	Every symmetric profile and
%	the congressional district profile
%	are special cases
%	of the profile $\phi$
%	in Lemma \ref{lem:cdm}.
%	\label{claim:profiles}
%\end{claima}
%
%\begin{proof}[Proof of Claim \ref{claim:profiles}]
%\end{proof}

\begin{claima}
	$\frac{s_n^2}{n}\to
	\int_{\underline{w}}^{\bar{w}}
	w^2\mathbb{E}[
	\phi(\Theta,w)^2]
	dG(w)$.
	\label{claim:s_n}
\end{claima}

\begin{proof}[Proof of Claim \ref{claim:s_n}]
	This holds since sequence
	$(\sigma_i^2)$ is bounded
	and
	the statistical distribution $G_n$ induced by
	$(w_i)_{i=1}^n$ converges weakly
	to $G$.
\end{proof}

\begin{claima}
	Conditions ($\alpha$) and ($\gamma$)
	in the LLT hold.
	\label{claim:alpha_gamma}
\end{claima}

\begin{proof}[Proof of Claim \ref{claim:alpha_gamma}]
	This immediately
	follows from the fact
	that sequence $(X_i)$ is bounded and
	$s_n^2\to\infty$.
	In particular,
	it is enough to define $\bar{x}$
	to be any finite number
	greater than $\bar{w}$.
\end{proof}

\begin{claima}
	Condition ($\beta$) in LLT holds.
	\label{claim:beta}
\end{claima}

\begin{proof}[Proof of Claim \ref{claim:beta}]
	Recall that
	$\phi$
	has the form
	\[
	w_i\phi(\theta_i,w_i)=
	h_1(w_i)h_2(\theta_i)+h_3(w_i){\rm\,sgn\,}\theta_i.
	\]

	Let $a_i=h_3(w_i)$.
	We first check that
	the sequence $(a_i)$
	satisfies
	the requirements in condition ($\beta$).
	First, $(a_i)$
	is bounded since $h_3$
	is bounded.
	Now,
	for any $i$ and any $\delta>0$,
	\begin{equation*}
	\begin{split}
	\mathbb{P}\{|X_i-a_i|<\delta\}
	&\geq\mathbb{P}\{|X_i-a_i|<\delta
	\text{ and }
	\Theta_i>0\}\\
	&\geq
	\mathbb{P}\left\{\left|w_i\phi(\Theta_i,w_i)-
	h_3(w_i){\rm\,sgn\,}
	\Theta_i
	\right|<\delta
	\text{ and }
	\Theta_i>0
	\right\}\\
	&=
	\mathbb{P}\{|h_1(w_i)h_2(\Theta_i)|<\delta
	\text{ and }\Theta_i>0\}.
	\end{split}
	\end{equation*}
	Letting $\bar{h}_1>0$
	be an upper bound
	of $|h_1|$
	and $\Theta$
	a random variable
	distributed
	as $\Theta_i$,
	the last
	expression
	has the following
	lower bound independent
	of $i$:
	\[
	\mathbb{P}\{
	|h_2(\Theta)|<
	\delta/\bar{h}_1
	\text{ and }
	\Theta>0
	\}>0,
	\]
	which
	is positive 
	by the assumptions
	on $h_2$
	and on the distribution
	of $\Theta$.

	Next we
	check the limit condition
	in ($\beta$).
	Recall that
	$A(t,\varepsilon)$
	is the union of
	intervals
	\[
	\left(\frac{\pi m+\varepsilon}{|t|},\,\frac{\pi(m+1)-\varepsilon}{|t|}\right),\,
	m=0,\pm1,\pm2,\cdots,
	\]
	restricted to
	$(-\bar{x},\bar{x})$,
	where we can choose $\bar{x}$ to
	be any number greater than $\bar{w}$.
	To prove the limit
	condition in ($\beta$),
	it therefore suffices to
	verify that one such interval
	contains $X_i-a_i$
	with probability bounded away from zero,
	for all groups $i$ in some sufficiently
	large subset of groups.
	To do this,
	note that if $\Theta_i<0$,
	then
	$X_i-a_i=h_1(w_i)
	h_2(\Theta_i)
	-2h_3(w_i)$.
	The assumptions on
	$h_2$ and on the distribution
	of $\Theta$
	imply that
	for 
	any $\eta>0$,
	there exists a set
	$O_\eta\subset
	[-1,0]$
	with $\mathbb{P}\{\Theta\in
	O_\eta\}>0$
	such that
	if $\Theta\in O_\eta$
	then 
	$|h_2(\Theta)|\leq\eta$.
	Therefore,
	\[
	\Theta_i\in O_\eta
	\text{ $\Rightarrow$ }
	X_i-a_i\in T_{w_i,\eta},
	\]
	where
	\[
	T_{w_i,\eta}:=
	[-2h_3(w_i)-\eta h_1(w_i),\,
	-2h_3(w_i)+\eta h_1(w_i)].
	\]
	Since $h_1$
	is bounded,
	we can make  $T_{w_i,\eta}$
	an arbitrarily small
	interval around $-2h_3(w_i)$
	by letting
	$\eta>0$ be
	sufficiently small.
	Moreover,
	since $h_3$
	is continuous
	and not a
	constant,
	we can find
	a sufficiently
	small interval
	$[\underline{v},\bar{v}]
	\subset[\underline{w},\bar{w}]$
	with
	$\underline{v}<
	\bar{v}$
	such that
	if $w_i\in[\underline{v},\bar{v}]$,
	then
	$-2h_3(w_i)$
	is
	between,
	and bounded
	away from,
	$\frac{\pi m}{|t|}$
	and $\frac{\pi(m+1)}{|t|}$
	for some integer $m$.
	Fix such 
	an interval
	$[\underline{v},\bar{v}]$ 
	and define
	\[
	I:=\{i | w_i\in[\underline{v},\bar{v}]\}.
	\]
	Then,
	for sufficiently
	small $\eta>0$
	and $\varepsilon>0$,
	we have
	$T_{w_i,\eta}
	\subset
	A(t,\varepsilon)$
	for all
	$i\in I$.
	Fixing such $\eta>0$
	and $\varepsilon>0$,
	it follows that
	\[
	\Theta_i\in O_\eta
	\text{ and }
	i\in I
	\text{ $\Rightarrow$ }
	X_i-a_i\in
	A(t,\varepsilon).
	\]
	This implies that
	\[
	\mathbb{P}\{X_i-a_i\in A(t,\varepsilon)\}
	\geq
	\mathbb{P}\{
	\Theta\in O_\eta
	\}=:p>0
	\text{ for all }
	i\in I,
	\]
	and hence
	\[
	\frac{1}{\log s_n}
	\sum_{i=1}^n
	\mathbb{P}\{X_i-a_i\in A(t,\varepsilon)\}
	\geq
	\frac{n}{\log s_n}
	\cdot
	\frac{\sharp\{i\in I |
		i\leq n\}}{n}
	\cdot
	p.
	\]
	As $n\to\infty$,
	the first
	factor
	on the right-hand
	side tends
	to
	$\infty$
	since $s_n$
	has an asymptotic
	order of $\sqrt{n}$.
	The second
	factor
	tends to
	$G(\bar{v})-
	G(\underline{v})>0$,
	which is positive
	since
	$G$
	has full support
	on $[\underline{w},
	\bar{w}]$.
	Therefore the left-hand
	side tends to $\infty$.
\end{proof}

\begin{claima}
	\textit{As $n\to\infty$, uniformly in $w_1\in[\underline{w},\bar{w}]$,}
	\begin{equation}
	2\int_0^1
	\theta_1
	\sqrt{2\pi n}
	\mathbb{P}\{\Omega_n(\theta_1,w_1)\}dF(\theta_1)
	\to
	\frac{2w_1\mathbb{E}[\Theta\phi(\Theta,w_1)]}{
		\sqrt{
			\int_{\underline{w}}^{\bar{w}}
			w^2\mathbb{E}[
			\phi(\Theta,w)^2]
			dG(w)
		}
	}.
	\label{eq:omega}
	\end{equation}
	By part (i) of Lemma \ref{lem:pareto},\footnote{It is
		easy to check that part (i) of Lemma \ref{lem:pareto}
		holds for rules $\phi(\cdot,w_i)$ that depend
		on weight $w_i$ as well.}
	the left-hand
	side of (\ref{eq:omega})
	is $\sqrt{2\pi n}
	\pi_i(\phi;n)$,
	and therefore
	Lemma \ref{lem:cdm} holds.
	\label{claim:lem_cdm}
\end{claima}

\begin{proof}[Proof of Claim \ref{claim:lem_cdm}]
	By Claims \ref{claim:alpha_gamma} and \ref{claim:beta}, we may apply the LLT 
	to obtain
	\begin{equation*}
	\sqrt{2\pi s_n^2}
	\mathbb{P}\{\Omega_n(\theta_1,w_1)\}
	\to
	2w_1\phi(\theta_1,w_1).
	\end{equation*}
	By Claim \ref{claim:s_n},
	this means that
	\begin{equation}
	\sqrt{2\pi n}\theta_1
	\mathbb{P}\{\Omega_n(\theta_1,w_1)\}
	\to
	\frac{2w_1\theta_1\phi(\theta_1,w_1)}{
		\sqrt{
			\int_{\underline{w}}^{\bar{w}}
			w^2\mathbb{E}[
			\phi(\Theta,w)^2]
			dG(w)
		}
	}.
	\label{eq:apply_llt}
	\end{equation}
	Letting $\theta_1=1$
	maximizes the left-hand side
	of (\ref{eq:apply_llt})
	with the maximum value $\sqrt{2\pi n}\mathbb{P}\{\Omega_n(1,w_1)\}$.
	This maximum value
	itself converges to a finite limit.
	Hence the expression
	$\sqrt{2\pi n}\theta_1
	\mathbb{P}\{\Omega_n(\theta_1,w_1)\}$
	is uniformly bounded for all
	$n$
	and $\theta_1\in[0,1]$.
	By the Bounded Convergence Theorem,
	\[
	2\int_0^1
	\theta_1 \sqrt{2\pi n}\mathbb{P}\{\Omega_n(\theta_1,w_1)\}dF(\theta_1)
	\to
	2\cdot
	\frac{2w_1\int_0^1\theta_1\phi(\theta_1,w_1)dF(\theta_1)}{\sqrt{
			\int_{\underline{w}}^{\bar{w}}
			w^2\mathbb{E}[
			\phi(\Theta,w)^2]
			dG(w)}}
	.
	\]
	Since $F$
	is symmetric
	and
	$\phi$
	is odd,
	this limit
	is exactly the one
	in (\ref{eq:omega}).

	To check the uniform convergence, note that
	for each $n$,
	the integral on the left-hand side
	of (\ref{eq:omega})
	is non-decreasing in $w_1$,
	since event $\Omega_n(\theta_1,w_1)$
	weakly
	expands as $w_1$ increases.\footnote{Let $\theta_1\in[0,1]$.
		If $\phi$ is a symmetric
		profile, {i.e.}, if $\phi(\theta_1,w_1)=\phi(\theta_1)$,
		then $w_1\phi(\theta_1)$ is non-decreasing in $w_1$.
		If $\phi=\phi^{\rm CD}$,
		then
		$w_1\phi^{\rm CD}(\theta_1,w_1)=
		c\,\text{sgn}(\theta_1) +(w_1-c)\theta_1$,
		which is non-decreasing in $w_1$ again.
		Thus event $\Omega_n(\theta_1,w_1)$ weakly expands as $w_1$ increases.}
	We have shown that
	this integral
	converges pointwise to
	a limit that is proportional
	to the factor $w_1\mathbb{E}[\Theta\phi(\Theta,w_1)]$,
	which is continuous in $w_1$.\footnote{If $\phi$ is a symmetric profile,
		this factor is linear in $w_i$.
		If $\phi=\phi^{\rm CD}$,
		the factor equals $c\mathbb{E}(|\Theta|)+
		(w_i-c)\mathbb{E}(\Theta^2)$,
		which is affine in $w_i$.}
	Therefore, the convergence in (\ref{eq:omega}) is
	uniform in $w_1\in[\underline{w},\bar{w}]$.\footnote{It is known
		that if $(f_n)$ is a sequence of non-decreasing functions
		on a fixed finite interval and $f_n$ converges pointwise to a
		continuous function, 
		then the convergence is uniform. See \citet{Buchanan1908}.}
\end{proof}

\subsection{Proof of Part (ii) of Lemma \ref{lem:pareto}}

This follows immediately from Lemma \ref{lem:cdm},
by noting that if $\phi$ is a symmetric profile, each group's rule
can be written as $\phi(\theta_j,w_j)=\phi(\theta_j)$.\qed

\subsection{Proof of Proposition 2}
\hspace{1cm}\linebreak
By part (ii) of Lemma \ref{lem:pareto},
we must show that
$\mathrm{Corr\,}[\Theta,\phi^a(\Theta)]$
is decreasing in $a\in[0,1]$.
By simple calculation,
\[
\mathbb{E}(\Theta^2)\cdot \mathrm{Corr\,}[\Theta,\phi^a(\Theta)]^2=
\frac{a\mathbb{E}(|\Theta|)+(1-a)\mathbb{E}(\Theta^2)}{
	a^2+2a(1-a)\mathbb{E}(|\Theta|)+(1-a)^2
	\mathbb{E}(\Theta^2)}.
\]
The derivative
of this expression
with respect to $a$
has the same sign as
\begin{equation*}
\begin{split}
&\Big\{\textstyle\frac{d}{da}
(a\mathbb{E}(|\Theta|)+(1-a)\mathbb{E}(\Theta^2))^2\Big\}\Big(a^2+2a(1-a)\mathbb{E}(|\Theta|)+
(1-a)^2\mathbb{E}(\Theta^2)\Big)\\
&-\Big(a\mathbb{E}(|\Theta|)+(1-a)\mathbb{E}(\Theta^2)\Big)^2
\Big\{\textstyle\frac{d}{da}
(a^2+2a(1-a)\mathbb{E}(|\Theta|)+(1-a)^2\mathbb{E}(\Theta^2))\Big\}\\
&=a(a\mathbb{E}(|\Theta|)+(1-a)\mathbb{E}(\Theta^2))(\mathbb{E}(|\Theta|)^2-\mathbb{E}(\Theta^2)).
\end{split}
\end{equation*}
This is negative for any $a\in(0,1]$, since
$\mathbb{E}(|\Theta|)^2\leq
\mathbb{E}(\Theta^2)$
in general,
and the full-support assumption
implies that this holds with
strict inequality.\qed

\subsection{Proof of Theorem
	\ref{thm:ineq}}

Clearly,
Lorenz dominance
is invariant under linear
transformations
of payoffs.
Thus,
it suffices to prove that
for large enough $n$,
the payoff profile
$\sqrt{2\pi n}\pi(\phi^{\rm CD};n)$
Lorenz dominates the payoff profile
$\sqrt{2\pi n}\pi(\phi;n)$.
By
equations
(\ref{eq:aw_phi})
and (\ref{eq:bw+c})
in the proof
of Theorem
\ref{thm:cdm}, as $n\to\infty$
these amounts converge to
$Bw_i+C$
and $A^{\phi} w_i$,
respectively.
A result by \citet[Proposition 2.3]{Moyes1994}
implies that
if $f$ and $g$
are continuous,
nondecreasing,
and positive-valued functions
such that 
$f(w_i)/g(w_i)$
is decreasing in $w_i$,
then
the distribution
of $f(w_i)$
Lorenz dominates
that of $g(w_i)$.
The ratio
$(Bw_i+C)/(A^{\phi}w_i)$
is decreasing in $w_i$,
and so the claimed Lorenz dominance
holds in the limit as $n\to\infty$.
Recalling that the convergences
are uniform,
the dominance holds for sufficiently large $n$.\qed

%%--------------------------------------SECTION-------------------------------------%%
%\section*{Acknowledgments}

%%--------------------------------------SECTION-------------------------------------%%

\phantomsection
\addcontentsline{toc}{section}{References}

\bibliography{bibbib}

%\section*{References}
%\begin{list}{}{\setlength{\itemindent}{-0.20truein}}

\end{document}